\documentclass[10pt,twocolumn]{IEEEtran} % Comment this line outIEEEtranN
\IEEEoverridecommandlockouts                              % This command is only
\usepackage{graphics,color} % for pdf, bitmapped graphics files
\usepackage{rotating}
\usepackage{hyperref}
\usepackage{algpseudocode}
\usepackage{graphicx}
\usepackage{algorithm}
\usepackage{cancel}
\usepackage{soul}
\usepackage{times}
\usepackage{comment}
\usepackage{algorithm,amssymb,url}
\usepackage{epsfig}
\usepackage[cmex10]{amsmath}
\usepackage{amsthm}
\usepackage{bm} 
\usepackage{dsfont}
\usepackage{venndiagram}
\usepackage{breqn}
\usepackage{textcomp}
\usepackage{titlesec}
\usepackage{booktabs}
\usepackage{multirow}
\usepackage{mathtools}
\usepackage{breqn}
\usepackage{mathtools}
\usepackage{multirow}
\usepackage{cancel}
\usepackage{subcaption}
\usepackage[numbers]{natbib} 
\titlespacing*{\section}{0pt}{1\baselineskip}{\baselineskip}
\titlespacing*{\subsection}{0pt}{1\baselineskip}{\baselineskip}
\newtheorem{theorem}{Theorem}
\newtheorem{remark}{Remark}
\newtheorem{definition}{Definition}
\newtheorem{corollary}{Corollary}
\newtheorem{lemma}{Lemma}

\DeclareMathOperator*{\argminA}{arg\,min}
\DeclareMathOperator*{\argmaxA}{arg\,max}

\algnewcommand\algorithmicforeach{\textbf{for each}}
\algdef{S}[FOR]{ForEach}[1]{\algorithmicforeach\ #1\ \algorithmicdo}

\renewcommand\algorithmicdo{}
\author{Mandar Datar}
\begin{document}

\newtheorem{assumption}{Assumption}
%%%%%%%%%%%%%%%%%%%%%%%%%%%%%%%%%%%%%%%%%%%%%%%%%%%%%%%%%%%%%%%%%%
%\title{Joint Pricing and Resources Allocation \\for 5G Network Slicing}
\title{Fisher Market Model based Resource Allocation for 5G Network Slicing 
}
% please modify the title as you want, Mandar
%%%%%%%%%%%%%%%%%%%%%%%%%%%%%%%%%%%%%%%%%%%%%%%%%%%%%%%%%%%%%%%%%%

\author{Mandar Datar\IEEEauthorrefmark{1}\thanks{Corresponding author:mandar.k.datar@gmail.com, this work was supported by Nokia Bell Labs and MAESTRO-5G-ANR. Part of the research described in this article has been included in the PhD thesis of the first author \cite{datar2022resource}, who was affiliated with INRIA Sophia Antipolis-Méditerranée and LIA at Avignon University in France. Presently, the author is associated with Orange Innovation, Châtillon, France}, Naresh Modina\IEEEauthorrefmark{2}, Rachid El Azouzi\IEEEauthorrefmark{4}, Eitan Altman\IEEEauthorrefmark{3}\IEEEauthorrefmark{4}, \\

\IEEEauthorrefmark{1}Orange Innovation, Châtillon, France
\IEEEauthorrefmark{2} (CEDRIC) CNAM - Paris, France

\IEEEauthorrefmark{4}CERI/LIA, University of Avignon, Avignon, France
\IEEEauthorrefmark{3}INRIA Sophia Antipolis-Méditerranée, France}

\maketitle
\begin{abstract}
  Network slicing (NS) is a key technology in 5G networks that enables the customization and efficient sharing of network resources to support the diverse requirements of next-generation services. This paper proposes a resource allocation scheme for NS based on the Fisher-market model and the Trading-post mechanism. The scheme aims to achieve efficient resource utilization while ensuring multi-level fairness, dynamic load conditions, and the protection of service level agreements (SLAs) for slice tenants. In the proposed scheme, each service provider (SP) is allocated a budget representing its infrastructure share or purchasing power in the market. SPs acquire different resources by spending their budgets to offer services to different classes of users, classified based on their service needs and priorities. The scheme assumes that SPs employ the $\alpha$-fairness criteria to deliver services to their subscribers. The resource allocation problem is formulated as a convex optimization problem to find a market equilibrium (ME) solution that provides allocation and resource pricing. A privacy-preserving learning algorithm is developed to enable SPs to reach the ME in a decentralized manner. The performance of the proposed scheme is evaluated through theoretical analysis and extensive numerical simulations, comparing it with the Social Optimal and Static Proportional sharing schemes.
\end{abstract}


\section{Introduction}
Communication technology has been playing an essential role in society’s digitalization and is
a significant contributor to a growing economy worldwide. Looking towards the future,
indeed, the next-generation wireless network is expected to grow and extend its  support in a whole new generation of applications like Augmented Reality (AR), Virtual Reality (VR) live broadcast, Internet of things (IoT), Autonomous driving, remote healthcare, automated manufacturing based on smart factories, etc.
A critical concern in integration of emerging sectors into a current wireless network is their heterogeneous and conflicting needs that the existing monolithic network is insufficient to meet. 
For example, automotive and healthcare applications require ultra-reliable services or extremely low latency,
whereas VR-live broadcast needs ultra-high-bandwidth communication.
\par
\par Several new concepts have been introduced for the upcoming 5G network design to satisfy these critical needs. Out of those, probably one of the most important ones is ``network slicing" \cite{GSMA}. It consists of virtualizing the physical resources and logically partitioning them with the help of technologies such as Software-Defined Networking (SDN) and Network Function Virtualization (NFV) \cite{ordonez2017network}. Each logical partitioned part is referred to as a slice and is tailored to meet the specific application/service needs. Generally, Infrastructure providers (InPs) own these resources and provide them to SPs through a dedicated slice. The slice-based provisioning is at the core of empowering SPs to manage the performance of their own dynamic and mobile user load locally. Network slicing enables slice tenants to share the same physical infrastructure flexibly and dynamically, which helps to utilize the resources more efficiently and economically. Though network slicing comes up with many advantages, it brings new challenges too. 
\par When networking architecture is based on shared resources, and SPs or tenants share a common infrastructure to support their service provision, the security and scalability of virtual networks are major concerns. Naturally, SPs request logical, independent, and isolated slices with complete service level agreements (SLA) protection. One of the simplest allocation strategies that offer SPs a guarantee of slice-level protection is static partitioning \cite{guo2013active}, where each resource required by the SPs is shared among them depending on their network stake or SLAs. However, this approach fails to provide load-driven flexibility when the service providers' user loads may vary with time and can be spatially inhomogeneous. In this regard, dynamic sharing of resources among SP users is one natural allocation choice that can meet the flexibility of the SPs  \cite{lee2018dynamic}\cite{altman2006survey}. Also, InPs want to maximize their return on investment by employing the dynamic sharing of resources as this lowers the capital cost and gives better resource utilization \cite{Nokia}. However, dynamic sharing of resources can expose the SPs to the risk of service-level agreement violation. Therefore, one of the critical concerns in 5G network slicing is how to efficiently and dynamically allocate limited resources to slice tenants with diverse characteristics and services while maintaining the protection of their SLAs. On top of that, most of the next-generation mobile applications/services such as AR and VR broadcast demand more data-intensive operations than those required by traditional mobile applications. Therefore, to cope with the requirements of additional computational power and memory resources for such services, cloud computing and edge computing are being integrated into the network architecture. As a result, unlike traditional mobile services where radio resources are the primary resource, a network slice is usually composed of heterogeneous resources, including radio access capacity or communication resources, edge storage memory, and computational resources etc. Thus it becomes even more challenging to design a multi-dimensional resources-sharing scheme which can ensure an ensemble of  user load-driven flexibility, protection of SLA and better network efficiency.

\par To address this issue, we propose a Fisher market (FM)-based resource allocation scheme, where market agents \emph{i.e.} SPs are assigned with fixed budgets or share of infrastructure according to their SLA. The InP sets the prices for the resources. Given prices announced by InP, the SPs distribute their budgets over different resources at different locations to procure the optimal bundle of heterogeneous resources required to support their services. In this work, we use market equilibrium solution approach to provide stable allocation and resource pricing. At the ME, the market is cleared, i.e., demand meets supply, and every agent is satisfied with allocated resources. To make the proposed resource allocation scheme practically viable, we implement it via the Trading post (TP) mechanism. This type of distributed approach protects the sensitive information of SPs and transfers each SP a direct control to tailor allocation by simply adjusting its bids. The required resources are allocated to SPs proportional to their bids. The proposed approach regulates the trade-off between efficient resource utilization and the degree of protection to SLA. On the one hand, it enables dynamic sharing, where tenants can redistribute their network share
based on the dynamic load; on the other hand, it also provides the SPs degree of protection by keeping the pre-assigned share intact throughout the allocation process.  
\subsection{Related work}
In this work, we model the resource allocation problem for 5G network slicing as a Fisher market where SPs act as consumers who purchase the different resources available at geographically distributed locations as goods. Computing the equilibrium to the FM is a challenging problem and has been the subject of much interest in the theoretical computer science community \cite{harrison1996computing}. Eisenberg and Gale in \cite{eisenberg1961aggregation} \cite{eisenberg1959consensus} and their generalization \cite{jain2007eisenberg} showed that if the utilities of agents in the market are the homogeneous function \footnote{A function is called as a homogeneous function of any degree ‘k’ if; when each of its elements is
multiplied by any number $t > 0$; then the value of the function is multiplied by $t^{k}$
.} of degree one. In that case, the market equilibrium can be found by solving a convex optimization programme, also widely known as Esenbeg-Gale (EG) program. It has been observed that the EG program also achieves proportional fairness \cite{kaneko1979nash} or optimum Nash social welfare  \cite{nash20164} among the market agents

\par An approach such as the EG program provides a centralized solution to find equilibrium; however, they do not represent the markets or the equilibrium concepts where agents practically interact with each other. Thus, the algorithmic game theory community has always been interested in designing algorithms that could plausibly describe the markets and the equilibrium concepts and allow agents to reach equilibrium \cite{roughga,devanur2008market,garg2015complementary,vazirani2011market,chen2016incentives}. For example, over a century ago, Walras \cite{walras1896elements,cheung2020tatonnement} proposed the most intuitive and natural algorithm, ``tâtonnement ", to find the equilibrium of a market. In this algorithm the price of resource increases with demand exceeding capacity and vice versa. The shortcomings of this type of approach are that first, they do not provide any causal relation between the prices and demand. In fact, price depend on agents' demand, and their demand depends on resource price. Second, total demand by agents may exceed capacity while applying the procedure, so managing the excess demand is critical.

\par The Shapley and Shubik \cite{tradingpost} sought to answer issues through the Trading post mechanism. The same mechanism has been discovered several times with different application domains, for example, the Kelly mechanism \cite{kelly1997charging} in computer networks, and proportional share scheme by Feldman et al. \cite{feldman2008proportional} in computer systems. The TP-mechanism provides an effective answer to many questions, however reaching equilibrium via the TP-mechanism is still challenging. Over the past decade, much attention has been given to designing algorithms to get the Fisher market equilibrium via TP-mechanism. Zang et al. in \cite{ZHANG20112691} showed that when CES\footnote{Constant elasticity of substitution (CES) utiltiy, $u(x_i)=\left(\sum_{j}a_{ij}{x_{ij}}^{\rho}\right)^{\frac{1}{\rho}},\text{if}\;0<\rho\leq 1$ signifies substitution relationship ;if $\infty\leq \rho\leq 0$ is complemtary relationship} utilities with substitute relationships determine resource demands of market agents, proportional response dynamics converge to FM equilibrium. Recently, in \cite{FisherDynamics}, Cheung et al. extended the above work to a case where any CES utility functions can determine market agents' demand and developed the distributed proportional dynamics to find the equilibrium of the FM.

\par The first step toward the multi-resource allocation problem in multi-server computing resources management was made in \cite{DRF1,DRF2,DRF3,DRF4,DRF5}. All these works proposed Dominant Resource Fairness (DRF) as criteria for multi-resource allocation. Recently, in \cite{Fossati} Fossati et al. studied the multiple resource allocation for network slicing under different fairness criteria such as Ordered weighted averaging (OWA), weighted proportional fairness (WPF), DRF and mood value rule. The authors proposed an Ordered weighted averaging (OWA) as fairness criteria. To balance the trade-off between inter-slice and intra-slice fairness, a new allocation criterion namely shared constrained slicing (SCS), was proposed by Zheng et al. in \cite{zheng2019elastic}.

\par 
In \cite{nguyen2018price}, Nguyen et al.  studied edge computing resource allocation problem for service as Fisher market model, where they only dealt with computation resources, considering that the linear function determines the agents' resource demand. Later in \cite{Nguyen2019AMF}, they extended the formulation to a multi-resource allocation problem by employing Leontief functions as agents' utilities. In the same vein, Moro et al. in \cite{Moro2020JointMO} cast resource allocation problem for 5G network slicing as FM, wherein apart from edge resources like computation and memory, authors also included the radio resource in the model. All three works mentioned earlier proposed the market equilibrium-based resource allocation as a solution and showed that the desired equilibrium-based allocation can be obtained by solving the EG program.
 \par In-network slicing context, similar to our current work, previously, in \cite{slicing_game}  Caballero et al. proposed the TP mechanism for bandwidth allocation problem. In their proposed scheme, tenants can customize their bandwidth demand by splitting their shares based on their geographically distributed user load. In advancement with the above work, in \cite{CaballeroTON2017} \cite{ZhengDeVeciana2018} Zheng et al. applied the same resource allocation scheme for statistical multiplexing of stochastic load. They showed that the resource allocation scheme induces a non-cooperative game, and the slices achieve efficient statistical multiplexing at the Nash equilibrium of the game. Further advancing on the same line, Caballero et al. in \cite{Guaranteed_Rate} introduced the admission control over users arrival to ensure the guaranteed service rate for slices' users. 

  \begin{table*}[]
    \centering
    \begin{tabular}{|c|c|c|c|c|c|c|}
    \hline
 \textbf{Article} & \textbf{Application} & \textbf{Fairness} &\textbf{Resources} & \textbf{Model} & \textbf{Learning} & \textbf{Equilibrium} \\
\hline
  \cite{slicing_game}\cite{CaballeroTON2017}& 5G slicing & $\alpha$-fairness & Radio & Trading Post & Best response & Nash Equilibrium \\
  \hline
 
 \cite{nguyen2018price,Nguyen2019AMF} & Fog Computing & None & Multi & Fisher Market & ADMM & Market Equilibrium \\
\hline

\cite{Moro2020JointMO} & 5G slicing & None & Multi & Fisher Market &  None & Market Equilibrium \\
\hline
 Our work & 5G slicing & $\alpha$-fairness & Multi & Fisher Market and Trading Post & Mirror decent & Market Equilibrium \\
\hline
       \end{tabular}
       \vspace{0.2cm}
   \caption{Literature review and research contribution positioning.}
    \label{tab:my_label}
\end{table*}
In literature, many works focused on developing distributed multi-resource resource allocation algorithms for network slicing. \citeauthor{Hassan} in \cite{Hassan} developed the distributed resource allocation schemes rooted in Kelly mechanism. In \cite{Panayotis} \citeauthor{Panayotis}  proposed an Alternating direction method of multipliers (ADMM) based distributed resource allocation mechanism for NS. Recently \citeauthor{Fossati_Disti} in
\cite{Fossati_Disti}, proposed decentralize 5G slice resource allocation schemes using cascade and parallel resource allocations. Above works differ from our present work as they didn't consider any budgets for system agents. Also, the last two works did not ensure any type of stability or equilibrium in the designed mechanisms. Finally, the resource allocation schemes in network slicing can be categorized into reservation-based \cite{reservationed_based1,reservationed_based2,reservationed_based3,reservationed_based4}, where resources are allocated on a reservation
basis and share-based resource allocations \cite{Moro2020JointMO,CaballeroTON2017,ZhengDeVeciana2018,Guaranteed_Rate}, which assign resources based on fixed overall shares associated to individual SPs. The advantage and disadvantages of these methods are well studied in \cite{budget_shared} 

 \par Our work is closely related to \cite{nguyen2018price,Nguyen2019AMF,Moro2020JointMO}, we also formulate the resource allocation problem as a Fisher market. However, this work departs from their works in following points. First, this work also incorporates end-user-level allocation and fairness in the model. Second, along the lines of \cite{CaballeroTON2017,ZhengDeVeciana2018,Guaranteed_Rate}, we design a distributed resource allocation scheme via a TP-mechanism that allows the SPs to reach the market equilibrium. However, above works only dealt with a single resource allocation problem. We generalize the mechanism for multiple resource type allocation. Finally, our work also extends the theoretical results from the \cite{FisherDynamics} by providing a TP-mechansim-based updating scheme to reach the ME of the Fisher market with complex utility functions.
\subsection{Main Contributions}
We list below the key contributions of our work in the paper.
\begin{enumerate}
\item In the context of network slicing, we formulate the system where the SPs need heterogeneous resources at geographically distributed locations to serve users from different service classes.
\item We cast the resource allocation problem for the aforementioned system as a Fisher market model and propose a market equilibrium as its stable solution.
\item We build a convex optimization programme whose optimal solutions provide ME for the formulated market.
\item We devise the bid updating rule vai TP mechanism that enables SPs to reach the ME in a decentralized fashion.
\item We investigate the efficiency and fairness properties of the proposed allocation scheme and perform a comparative analysis with two baseline allocation schemes: social optimal allocation and static proportional allocation schemes.
\end{enumerate}
The rest of the paper is organized as follows: Section \ref{sysytem model} introduces the system model. In Section \ref{Problem_formulation}, we cast the resource allocation problem as the Fisher market model. In Section \ref{CentralizedResAlloc} and Section \ref{ResAllocdecentra}, we provide centralized and decentralized approaches respectively to compute the equilibrium of the formulated market. Section \ref{sec:potetialfunction} is specially dedicated for developing a potential function, which is needed for developing a decentralized allocation scheme. In Section \ref{bidupdating}, we provide bid updating rule which allows SPs to reach the desired market equilibrium. In Section \ref{fairnessefficiency}, we investigate the fairness and efficiency properties of the proposed allocation scheme. In section \ref{numerical}, we validate the performance of the proposed allocation scheme with extensive numerical simulations. Finally, we conclude the paper by summarizing the results and future work.

\begin{table}
%\begin{center}% used the environment to augment the vertical space
% between the caption and the table
\centering
\begin{tabular}{l c p{0.5\columnwidth}}
\toprule
$\mathcal{C}:=\{1,\ldots,C\}$ & $\triangleq$ & set of base stations or cells\\
$\mathcal{S}:=\{1,\ldots,S\}$ & $\triangleq$ & set of SPs (tenants) \\
$\mathcal{R}_{c}$ & $\triangleq$ & set of resources at base station $c$  \\
$\mathcal{K}^s$ & $\triangleq$ & set of user class supported by SP $s$\\
$\mathcal{K}^s_{c}$ & $\triangleq$ & set of user class supported by SP $s$ at cell $c$\\
$B_{s}$ & $\triangleq$ & budget or network share of SP $s$ \\
$x^s_{ckr}$ & $\triangleq $ & amount of resource type $r$ allocated to SP $s$ for users from class (slice) $k$ at cell $c$\\
$x^s_{ck}$ & $\triangleq $ & spending by SP $s$ on users from class $k$ at cell $c$\\
$b^s_{ckr}$ & $\triangleq $ & spending by SP $s$ for users from class $k$ on resource type $r$ at cell $c$\\
$b^s_{ck}$ & $\triangleq $ & spending by SP $s$ on users from class $k$ at cell $c$\\

${n}^s_{ck}$ & $\triangleq$ & number of user from SP $s$ belonging to class $k$ at cell $c$\\
$D_{k} =\left ( d_{k1}\dots d_{kR} \right )$& $\triangleq$ & base demand vector for user class $k$\\
$d_{kr}$& $\triangleq$ & the minimum amount of resource type $r$ needed by a user from class $k$ to achieve unit service rate\\
$u_{\nu}$& $\triangleq$ & service rate experienced by user $\nu$\\  
$u^{s}_{ck}$& $\triangleq$ & total service rate experienced by users belonging to class $k$ in cell $c$ \\  

$U_s$ & $\triangleq$ & utility of SP s \\
$\alpha_s $& $\triangleq$ & $\alpha$ fairness parameter for SP $s$\\
$p_{cr}$& $\triangleq$ & price per unit of resource type $r$ at cell $c$.\\ 
$\Phi(b)$ & $\triangleq$ & potential function to Esenberg-Gale Program\\
$\Psi (x)$ & $\triangleq$ & objective function of Esenberg-Gale Program\\
$\Upsilon(p)$& $\triangleq$ & dual function of Esenberg-Gale Program\\
\bottomrule
\end{tabular}
\vspace{0.2cm}
\caption{Main notations used throughout the paper}
\label{tab:TableOfNotationForMyResearch}
\end{table}
\section{System Model}\label{sysytem model}
\begin{figure}
    \centering
    \includegraphics[width=0.9\columnwidth]{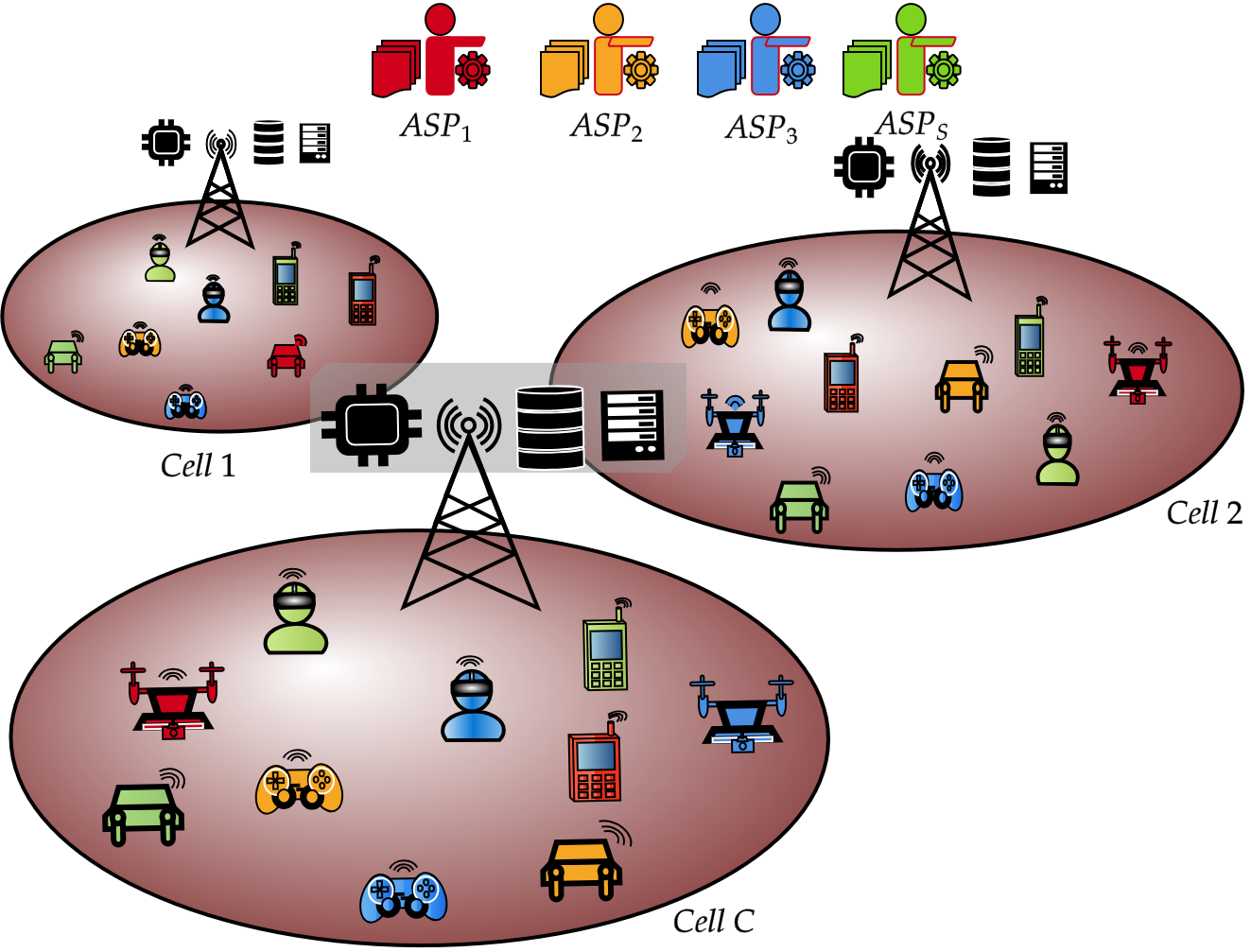}
    \caption{The SPs support the different applications at various locations through dedicated slices }
    \label{fig-system1}
\end{figure}
\begin{figure*}
\centering
\includegraphics[width=0.9\columnwidth]{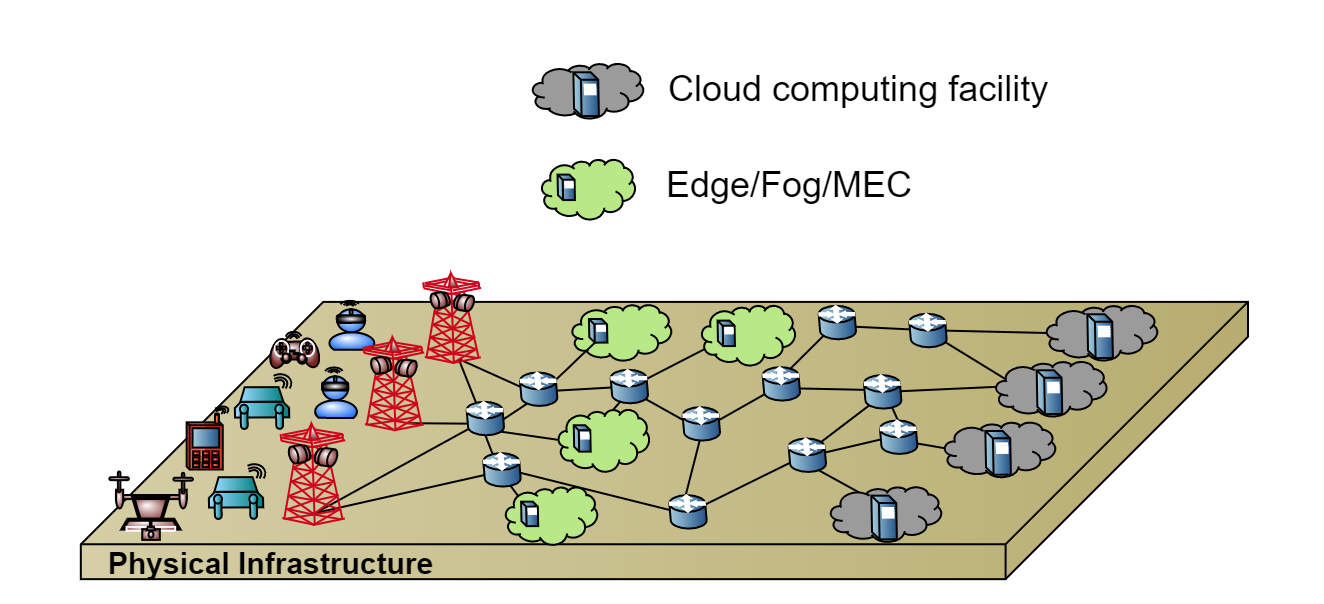}
\includegraphics[width=0.9\columnwidth]{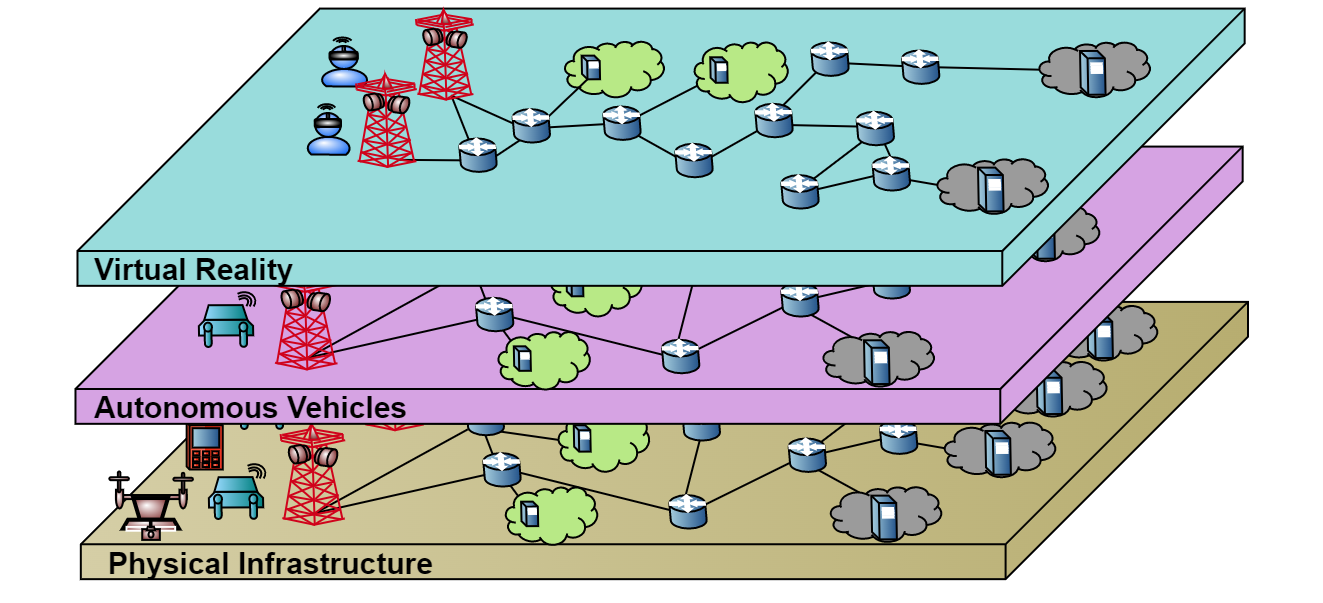}
\caption{ \cite{9927325} The SPs support the different applications through dedicated slices sharing the same physical infrastructure. }
\label{fig-system}
\end{figure*}
We consider a system with a set of InPs, who own the physical resources such as CPU, memory, radio resource, etc., at a geographically distributed set of locations or cells (macro, micro, small)\footnote{ We use term `cell' and `location' in this work interchangeably.} as represented in Fig.~\ref{fig-system}. Let  set of cells be denoted by $\mathcal{C}=\lbrace1,\ldots, C\rbrace$ and $\mathcal{R}^{c}=\lbrace1,\ldots, R\rbrace$ represents set of resources available at each location $c$. A set of SPs denoted by $\mathcal{S}=\lbrace1,\ldots, S\rbrace$ lease available resources from InPs through network slicing to create one or more slices to provide heterogeneous vertical applications (services) (e.g., IoT, VR, online gaming, autonomous driving and healthcare etc.) for the subscribed users at multiple locations. Here network slicing is a process in which the physical network is virtualized and then logically separated to match the SP requirements. For instance, a SP can lease some resources and create two slices, one for VR applications and another for autonomous driving applications. We consider that the resource requirement of SPs in each cell depends on the number of users in that cell and the type of service they provide.
\subsection{User utility model}
As shown in Fig.~\ref{fig-system1}, a set of users $\mathcal{V}$ are categorized into a set of classes or applications denoted by $\mathcal{K}$, where each class represents a different service requirement. Users obtain the resources by subscribing to the services offered by various SPs. The SP needs to provide heterogeneous resources to meet the service rate defined per class. We consider that every user needs a minimum allotment of each resource to meet a certain service rate. Let $D_{k} =\left ( d_{k1}\dots d_{kR} \right )$ be the base demand vector, where the element $d_{kr}$ is the amount of resource type $r$ required by a user of class $k$ to achieve a unit service rate. The service rate obtained by the user that belongs to class $k$ for acquiring a resource bundle $\left( x_{k1}\dots x_{kR}\right)$ is defined by the following Leontief function \cite{leontif},
\begin{equation}
u_{\nu}=\underset{r}{min}\left \{ \frac{x_{kr}}{ d_{kr}} \right \}\label{leontif}
\end{equation}
where $x_{kr}$ is the amount of resource type $r$ allocated to the user of class $k$. For instance, consider that a user from class $k$ with its base demand vector defined as  $d_{k,BW}=0.2$ and $d_{k,CPU}=0.1$ units, receives $0.4$ units of bandwidth and $0.2$ units of CPU then the service rate achieved by user $v$ is given by , $u_{v}=min\left \{ \frac{0.4}{0.2} ,\frac{0.2}{0.1}\right \}$. Observe that increment in the bandwidth to $0.6$ units does not increase the utility which highlights the main attribute of the Leontief function, i.e., improvement in the utility is possible only with proportional increment in all the allocated resources.
\subsection{Service Provider Utility Model}
Let us consider that each SP supports various classes or applications of users at a given location $c$ and this support may vary based on the location. Let a set of user classes supported by a SP $s$ denoted by $\mathcal{K}^s$, and for each service type, SPs operate through a separate slice. Let $n^{s}_{kc}$ be the number of users from class $k$ associated to SP $s$ present in cell $c$, and $\mathcal{K}^s_{c}$ be the set of classes supported by SP $s$ at location $c$. We assume that each SP treats the users in the same class uniformly, i.e., the service rate provided by the SP $s$ is equally divided among the users in the same class at location $c$, in general, this is valid for all locations. Thus a class-level allocation can easily be treated as to user-level allocation. The utility that accounts for  the benefit of resource allocation obtained by the SP is defined as follows
\begin{equation}
U_s=\sum_{c\in\mathcal{C}}\sum_{k\in\mathcal{K}^s_c}\frac{n^s_{ck}}{(1-\alpha_s)}\left(\frac{u^s_{ck}}{n^s_{ck}}\right)^{(1-\alpha_s)}
\end{equation}
In the above equation $u^s_{ck}=\min_{r\in\mathcal{R}_{c}}\left \{  \frac{x^s_{ckr}}{d^s_{ckr}}\right \}$, where $u^s_{ck}$ and $x^s_{ckr}$  represents the total service rate and the amount of resource type $r$ allocated to set of users belonging to class $k$ in cell $c$ respectively and $d^s_{ckr}=d_{kr}$. For the sake of convenience, we replace $\left( n^s_{ck} \right)^{\alpha_s}$ with $w^s_{ck}$ and write the above utility as 
\begin{equation}
U_s=\sum_{c\in\mathcal{C}}\sum_{k\in\mathcal{K}^s_c}\frac{w^s_{ck}}{(1-\alpha_s)}\left(u^s_{ck}\right)^{(1-\alpha_s)} \label{sputility}
\end{equation}
The utility \eqref{sputility} demonstrates that the SPs aim to attain the well-known $\alpha$-fairness criteria \cite{alpha_fair} among the classes of users while delivering the service at different cells. The values of $\alpha\in  \left [  0 , 1\right )\cup \left (1 ,\infty \right ]  $ interpolate between individual fairness among the users and efficiency of the system. The $\alpha=0$ corresponds to the utilitarian (average) objective where the goal is to optimize system efficiency, while $\alpha\rightarrow \infty$ corresponds to max-min fairness (the egalitarian objective). The $\alpha = 1$ and $\alpha = 2$ corresponds to the widely established proportional fairness and potential criterion, respectively.
\begin{equation}
U_s=\begin{cases}
\sum_{c\in\mathcal{C}}\sum_{k\in\mathcal{K}^s_c}w^s_{ck}\left(u^s_{ck}\right) & \text{ if } \alpha_s=0\\
\prod_{ck} \left(u^s_{ck} \right )^{w^s_{ck}} & \text{ if } \alpha_s= 1,\\
 \min_{c,k} \left \{ \frac{u^s_{ck}}{n^s_{ck}} \right \}  &\text{ if } \alpha_s= \infty\\ 
\sum_{c\in\mathcal{C}}\sum_{k\in\mathcal{K}^s_c}\frac{w^s_{ck}}{(1-\alpha_s)}\left(u^s_{ck}\right)^{(1-\alpha_s)} & \text{ if } \alpha_s= \text{ others }
\end{cases}
\end{equation}
Here, alpha fairness offers service providers a flexible approach to resource allocation. By adjusting the alpha parameter, SPs can strike a suitable balance between optimizing system performance and promoting fairness. When $\alpha_{s} = 0$, it prioritizes system efficiency, maximizing the overall system throughput or weighted sum of utilities. However, this can lead to highly unequal distributions of resources among users. For instance, when the weight $(w^s_{ck})$ of a particular class is higher than the other, it can result in that class receiving significantly more or even all resources than others, leading to potential dissatisfaction for remaining classes. When  $\alpha_{s} = 1$, it focuses on fairness and ensures a proportional distribution of resources among all classes according to their weights,  where none of the service classes can receive a zero service rate.

\subsection{Service Provider Budgets}
We further consider that each  SP $s \in \mathcal{S}$ is allocated with a finite budget or share $B_{s}$ that represents its share of total infrastructure such that $\sum_{s\in\mathcal{S}}B_{s}=1$, the budget allocated for each SP depends on its service level agreement with the InP. For example, InPs can allocate different budgets (virtual budgets) to the SPs depending on their priorities, initial investment in the infrastructure and/or potential revenue generation.
\begin{remark}
In the above model, we have considered that SPs treat users uniformly at each cell while delivering service to users from the same class. However, our framework is more general and applicable to cases where SPs may handle each user independently and provide the service according to their preferences and priorities; in such cases, each user will be treated as a separate class.       
\end{remark}
\begin{remark}
We use the Cobb-Douglas function instead of the logarithmic function, generally used in $\alpha$ fairness criteria; however, both the operations perform the identical task of attaining proportional fairness among the users.
\end{remark}

\subsection{Service Providers' Resource Allocation Problem}
Each SP needs to pay for the resources he leases from the InP. Let us consider $p_{c}=\left(p_{cr},\dots,p_{cR}\right)$ as the vector of prices charged at cell $c$ where the $p_{cr}$ is the price per unit of resource type $r$ at cell $c$.
Given the prices charged by the InP for resources, we anticipate the SPs to act as rational agents and spend their budgets to procure the resources in a manner that maximizes their utility. The decision problem for each SP $s$ to find an optimal bundle of resources to be requested to the InP is defined by the following optimization problem.

\begin{subequations}
\begin{align}
P_{s}: \hspace{10pt}& \underset{{\bf x,u}}{\text{Maximize :}} 
& & \sum_{c\in\mathcal{C}}\sum_{k\in\mathcal{K}^s_c}\frac{1}{1-\alpha_s}w^s_{ck}\left(u^s_{ck}\right)^{1-\alpha_s}\label{maxutility}\\
& \text{subject to}   & &u^s_{ck}\leq \frac{x^s_{ckr}}{d^s_{ckr}} \quad \forall c\in {\cal C}_s,k\in {\cal K}^s_c,r\in {\cal R}_{c}, \label{linearsputiltiy}\\
& & &  \sum_{ c \in {\cal C}}\sum_{k\in\mathcal{K}^s_c}\sum_{r\in\mathcal{R}_{c}} p_{cr} x^s_{ckr}    \leq B_{s} \label{budgetconst}.
\end{align}\label{marketprob}
\end{subequations} \label{SPsopti}

Here, the objective is to find an optimal allocation vector $[x^s_{ckr}]$ that solves the problem ~\eqref{maxutility}-\eqref{budgetconst}. Such optimal share of the resources are allocated to the SP through network slicing process, i.e., each SP is assigned with a slice that contains the resources that were divided optimally among the competing SPs. As mentioned before, the use of $\alpha_s$ ensures the fairness criteria among the users of different classes that are associated to the SP $s$. Constraints defined in eq.~\eqref{linearsputiltiy} ensure that the resulting service rate $u^s_{ck}$ does not consider the excess resources but only the resources that are in proportion. Whereas constraints in \eqref{budgetconst} guarantee that the resulting allocation adheres to the budget limitations of the SP.

\section{Resource Pricing and  Equilibrium Problem}\label{Problem_formulation}

We  consider that the capacity of the resource released by the InP in each location is finite. Given the per-unit prices for resources set by the InP, the total resource requested by the SPs through the optimization program \eqref{marketprob} may violate the infrastructure capacity. 
In this work, our primary goal is to design a resource pricing and allocation mechanism for network slicing, which maximizes the network resource utilization and assigns SPs their favourite bundle while adhering to the capacity limit of each resource type as well the budget constraints of the SP. In market economics, this pricing problem is addressed as a market clearing problem, where the market prices are determined in such a way that the amount of resources requested by the buyers matches the amount of resources supplied by the sellers. Thus to address the issue, we adopt a market clearing approach inspired by market economics and formulate the resource allocation problem as a Fisher market where $S$ SPs act as consumers who spend their budget to purchase the resources available at the different cells. At the same time, InPs operate as producers selling capacity-constrained resources in exchange for consumer spending. Now we define the market formally as
\begin{definition}
$\mathcal{M}:=\left \langle \mathcal{S},\left({B}_{s}\right)_{s\in\mathcal{S}},\bigcup_{c\in\mathcal{C}}\mathcal{R}_{c},\left(U_{s}\right)_{s\in\mathcal{S}},p\right \rangle $ as follows:
\begin{itemize}
\item[•]Player set: the set of SPs  $\mathcal{S}$
\item[•]Budgets :$B_{s}$
\item[•]Resources set: $\bigcup_{c\in\mathcal{C}}\mathcal{R}_{c}$
\item[•]Utility: The utility of each SP $s$ is equal to the $U_{s}$
\item[•]Price vector: $p$
\end{itemize}
\end{definition}
\par The primary goal of this work is to provide a Fisher market-based resource allocation scheme that effectively prices and allocate limited physical resources to SPs with heterogeneous requirements and preferences. Since in the Fisher market, resources are allocated to SPs depending on their rational decisions and interactions, we investigate the outcome of the proposed FM model in light of game theory and recall the definition of market equilibrium. 
\begin{definition}
 A market equilibrium (ME) is defined as the prices and resultant $(\hat{\mathbf{p}},\hat{\mathbf{x}})$ allocation, where the market clears its resources and SPs get their favourite resource bundle. Mathematically $(\hat{\mathbf{p}},\hat{\mathbf{x}})$ is ME if following two conditions are satisfied.
\begin{enumerate}
\item Given the resource price vector, every SP spends its budget such that it receives resource bundle $\hat{\mathbf{x}}$ that maximizes its utility. 
\begin{equation}
\begin{aligned}
\forall s ,\hat{x}_s \in \argmaxA_{\sum_{c}\sum_{k}\sum_{r}{x^s_{ckr} \hat{p}_{cr}}\leq B_{s}}U_s(x)
\end{aligned}
\end{equation}
\item Every resource is either
fully allocated or has zero price, i.e., we have:
$(\sum_{s}\sum_{k}\hat{x}^s_{ckr}-C_{cr})\hat{p}_{cr}=0$
\end{enumerate}
 
\end{definition}
In the above definition, the first condition is meant that the equilibrium allocation maximizes satisfaction or the return of market investment of each SP; no equilibrium could be otherwise established. 
The second condition represents Walras's law \cite{walras1896elements}, which means that either the total demand of each resource meets the capacity and will be positively priced; otherwise, that resource is allocated free of cost. Undoubtedly, zero priced resources can be allocated arbitrarily to SPs. However, an additional allocation of these resources will not increase their satisfaction level. 
%\par The proposed resource allocation scheme aims to obtain resource allocation and prices functioning as market equilibrium, where every slice is satisfied with their assignment, and all resources get completely utilized. 
\begin{assumption}
In this work, we consider that for each SP, its user load  is dynamic and spatially inhomogeneous. The shares (budgets) are allocated to SPs over a long timescale (months/days), and the distribution of shares by the SPs over the resources is performed at a fast time scale (minutes/seconds), such that user load is regarded as stationary while performing the allocation.
\end{assumption}
In next section we will delve into the algorithmic procedures and mathematical techniques utilized to compute the general equilibrium of a Fisher market. We will investigate solutions that predominantly utilize either a centralized approach or a decentralized approach.

\section{Centralized approach}\label{CentralizedResAlloc}
In the centralized approach, the Infrastructure Provider holds information about the utility functions and base demands of service providers. Leveraging this knowledge, the InP aims to optimize resource allocation by computing appropriate prices and quantities of resources for each SP. The primary objective is to ensure that, collectively, these allocations satisfy the equilibrium (ME) conditions.
In order to solve the market equilibrium problem, we first introduce the concept of the Esenberg-Gale (EG) optimization problem, which can then be used to find the equilibrium of the Fisher market $\mathcal{M}$ under some conditions. Suppose all utilities of consumers in the FM are concave and homogeneous of degree one. In that case, according to \cite{roughga}\cite{eisenberg1961aggregation}, the equilibrium of that market can be obtained by solving the below Esenberg-Gale optimization program.
 \begin{subequations}\label{galeconvex}
\begin{align}
& \underset{{\bf x,u}}{\text{Maximize :}} 
& & \sum_{s\in {\cal S}}B_s \log(U_s) \label{primal}\\
& \text{subject to} & &  U_s=\left(\sum_{c\in\mathcal{C}}\sum_{k\in\mathcal{K}^s_c}w^s_{ck}\left(u^s_{ck}\right)^{1-\alpha_s}\right)^{\frac{1}{1-\alpha_s}} \;\forall \; s\in {\cal S}, \label{const1}\\
&  & &u^s_{ck}\leq \frac{x^s_{ckr}}{d^s_{ckr}} \quad \forall \ s\in {\cal S},c\in {\cal C}_s,k\in {\cal K}_c,r\in {\cal R} \label{const2},\\
& & &  \sum_{ s \in {\cal S}}\sum_{k\in\mathcal{K}^s_c}  x^s_{ckr}    \leq C_{cr}, \quad \forall \ c\in {\cal C},r\in {\cal R} \label{const3}.
\end{align}
 \end{subequations}
Suppose $\mathbf{x}^{*}$ and $\mathbf{p}^{*}$ be the primal and dual optimal solutions of EG program, where $\mathbf{p}^{*}$ be the dual variable or lagrangian multiplier associated with capacity constraints then $\mathbf{x}^{*}$ and $\mathbf{p}^{*}$ represent the market equilibrium allocation and prices of market $\mathcal{M}$.    
\begin{theorem}
Optimal allocation $\mathbf{x}^{*}$ and the corresponding prices $\mathbf{p}^*$ (dual variable associated to \eqref{const3} of optimization program represents the market equilibrium
\end{theorem}
\begin{proof}
In the proposed FM model utility of each SP is concave and homogeneous of degree one; hence, the result follows from \cite{eisenberg1961aggregation}\cite{roughga}. For completeness, we present an elementary, self-contained step wise proof
\begin{enumerate}
\item First, we show that the optimal solution $(\mathbf{x}^{*})$ to program and its dual solution $(\mathbf{p}^{*})$ pair is budgets balancing
\emph{i.e,} at the solution pair the budget of each SP is fully utilize $$\forall s\in\mathcal{S}\quad\sum_{ c \in {\cal C}}\sum_{k\in\mathcal{K}^s_c}\sum_{r\in\mathcal{R}_{c}} p^*_{cr} {x^s_{ckr}}^*=B_{s}$$
the detailed proof is provided in the Appendix \ref{proofmarketequi}
\item Given price vector  $(\mathbf{p}^*)$ we find best response (demand) by each SP $s$ by solving \eqref{marketprob}, (the derivation of best response is provided in Appendices \ref{bestresponse}) and it gives the same solution allocation $\mathbf{x}^{*}$
or consider any possible allocation $\mathbf{y}$ that can be requested by SP such that 
 $$\forall s\in\mathcal{S}\quad\sum_{ c \in {\cal C}}\sum_{k\in\mathcal{K}^s_c}\sum_{r\in\mathcal{R}_{c}} p^*_{cr} y^s_{ckr}\leq B_{s}$$
 \item Now, we can show that $\forall s\in\mathcal{S},\quad U_{s}(y_s)\leq U_{s}(x^*_s)$
\end{enumerate}
which proves that $(\mathbf{x}^*,\mathbf{p}^*)$ solution to program, is  market equilibrium.
\end{proof}
We have already discussed in the previous section, the market equilibrium solution to the proposed problem \eqref{marketprob} can be found as a solution to the equivalent convex optimization problem \eqref{galeconvex}, which can then utilize by InP to implement the proposed allocation scheme. However, such a centralized implementation requires all the SPs' private utility functions to be available to InP. This is rarely possible, as it is generally not acceptable for SPs to reveal their private data to third parties. In the next section section we provide with decentralized resource allocation approach to overcome this issue
\section{Decentralized approach}\label{ResAllocdecentra}
 In this section, we focus on developing the decentralized algorithm enabling the service provider to reach the market equilibrium of $\mathcal {M} $ without disclosing their private utilities. One of the possible solutions in this direction is to solve the convex optimization problem using a Walras' t$\hat{\text{a}}$tonnement-like algorithm where the resource price is raised if the demand for a resource exceeds the resource supply and decreases if the demand for the resource is less than the supply. However, this is not the way how the market generally functions in practice; this type of approach does not always guarantee the ability to satisfy the resource capacity while applying a process. \par To deal with this issue, we advocate an alternative approach by Shapley and Shubik, well known as the Trading post mechanism. The proposed method does not require the service providers to reveal their utilities; SPs can distribute their budgets over the resources and customize their allocations. In TP mechanism, SPs distribute their budget over their required resources in bids. Once all SPs place the bids, each resource type's price is determined by the total bids submitted for that resource. Let SP $s$ submits a bid $b^s_{cr}$ to resource $r$ at cell $c$. The price of resource type $r$ at cell $c$ is then set to $\sum_{s\in \mathcal{S}} b^s_{cr} $, accordingly SP $s$ receives a fraction of $x^s_{cr}$ in return to his spending of $b^s_{cr}$ 
\begin{equation}
x^s_{cr}=\begin{cases}
 \frac{b^s_{cr}}{\sum_{s'\in\mathcal{S}} b^{s'}_{cr}}& \text{ if } b^{s}_{cr}>0, \sum_{s'\in\mathcal{S}} b^{s'}_{cr}>0\\ 
 0& b^{s}_{cr}=0,\sum_{s'\in\mathcal{S}} b^{s'}_{cr}>0\\
\text{as per demand}& \sum_{s'\in\mathcal{S}} b^{s'}_{cr}=0
\end{cases}
\end{equation}
In our framework, the same resource is required by users from different classes; thus, the total budget spent by SPs  $s$ on resource  $r$  is the sum of budgets spent  by the SP on resources $r$ for the set of its users belonging to all possible classes  $b^s_{cr}=\sum_{k}b^s_{ckr}$ .
We assume that the SPs are price takers\footnote{Price taker vs a price maker,  price takers accept the market price as given, while price makers have some degree of control over the price due to their market power.}, and they request different amounts of the resources
by distributing their budgets over the resources in bids. The
InP announces the resources’ prices and allocates the resources according to TP-mechanism. If all the SPs are satisfied with the allocation and prices announced by InP, the mechanism has reached ME; otherwise, SPs might modify their distribution of budgets (demand) depending on the current prices. This brings new the challenge of dynamics or bids updating scheme: (how) do SPs reach a market equilibrium via the TP mechanism ? In the coming sections, we focus on developing the bid updating rule that enables the SPs to reach the ME of the proposed allocation scheme. However, before moving directly to the main results, we build some mathematical tools that will be required afterwards in developing the bid updating rule and prove its convergence results.
%The TP mechanism allows for fair resource allocation as each allocated $d^i_{cr}$ is proportional to its bid $b^i_{cr}$. Each MVNO $i$ has a budget-constrained $B_{i}$ and the total bids of MVNO $i$  must be sump to its budget, $\sum_{c\in\mathcal{C}}\sum_{r\in\mathcal{R}} b^i_{cr}=B_i$. 
%Note here budget-constrained makes this mechanism different from the well know Kelly mechanism \cite{Kelly1998}, where bids are unconstrained, and utility is defined as the benefit minus the cost for resources.
 \subsection{Potential function}\label{sec:potetialfunction}
In this section, we construct a potential function to the Eisenberg-Gale program \eqref{galeconvex} and show the optimal solution to the problem (\emph{i.e.} ME) is an optimal point of the candidate potential function. In this article, we restrict our analysis to the case when the $\alpha_{s}$ the alpha fairness criteria employed by each SP $s$ takes value in  $\left[1,\infty\right]$, for the remaining case when $\alpha_{s}\in \left[0, 1\right]$  requires complex saddle point analysis and we keep this for future work. Next, we show that when all the SP employ the $\alpha$ fairness criteria with $1\leq\alpha_s\leq\infty$, our designed potential function is convex, and its minimal point represents the ME. Now to start with the designing of the potential function, we consider the dual of an optimization problem where the goal is to minimize $\Upsilon (p)$
\begin{dmath}
\Upsilon (p)= \max_{\sum_{c,k,r}x^s_{ckr}p_{cr}=B_s}\left ( \sum_{s}B_s \log\left ( U_s \right )+\sum_{c\in\mathcal{C}}\sum_{r\in\mathcal{R}_c} p_{cr} \left (1-\sum_{s}\sum_{k\in\mathcal{K}^s_c}x^s_{crk}\right )\right ) 
\end{dmath}  
Now we introduce the potential function and write separately depending on the parameter value $\alpha_s$ used by the service providers. Let $\Phi(b_{\#})$ denote the potential function of the EG-problem when the alpha fairness parameter $\alpha$ with the condition $\#$  has been applied by the SPs. For example, $\Phi(b_{\geq})$ denotes the potential function when SPs apply the fairness with $\alpha\geq1$. We use the same notation for the remaining article to describe the potential function and its connection with the SPs' $\alpha$ fairness parameter.   

%\begin{align}
%\Phi_{>0}(b)=-\sum_{s\in\mathcal{S}}\sum_{c\in\mathcal{C}^s}\sum_{k\in\mathcal{K}^s_{c}}\sum_{r\in\mathcal{R}_{c}}b^s_{ckr}\log\left(\frac{b^s_{ckr}}{p_{cr}d^s_{ckr}}\right)-\frac{1}{(1-\alpha_s)}\sum_{s\in\mathcal{S}}\sum_{c\in\mathcal{C}^s}\sum_{k\in\mathcal{K}^s_{c}}b^s_{ck}\log\left(\frac{b^s_{ck}}{w^s_{ck}}\right)
%\end{align}
%\begin{align}
%\Phi_{=0}(b)=\sum_{s\in\mathcal{S}}\sum_{c\in\mathcal{C}^s}\sum_{k\in\mathcal{K}^s_{c}}\sum_{r\in\mathcal{R}_{c}}w^s_{ck}b^s_{ckr}\log\left(\frac{b^s_{ckr}}{p_{cr}d^s_{ckr}}\right)
%\end{align}
%\begin{align}
%\Phi_{<0-\{-\infty\}}(b)=\sum_{s\in\mathcal{S}}\sum_{c\in\mathcal{C}^s}\sum_{k\in\mathcal{K}^s_{c}}\sum_{r\in\mathcal{R}_{c}}b^s_{ckr}\log\left(\frac{b^s_{ckr}}{p_{cr}d^s_{ckr}}\right)-\frac{1}{(1-\alpha_s)}\sum_{s\in\mathcal{S}}\sum_{c\in\mathcal{C}^s}\sum_{k\in\mathcal{K}^s_{c}}b^s_{ck}\log\left(\frac{b^s_{ck}}{w^s_{ck}}\right)
%\end{align}
%\begin{align}
%\Phi_{-\infty}(b)=\sum_{s\in\mathcal{S}}\sum_{c\in\mathcal{C}^s}\sum_{k\in\mathcal{K}^s_{c}}\sum_{r\in\mathcal{R}_{c}}b^s_{ckr}\log\left(\frac{b^s_{ckr}}{w^s_{ck}p_{cr}d^s_{ckr}}\right)
%\end{align}

%%%%%%%%%%%%%%%%%%%%%%%%%%%%%%%%%%%%%%%%%%%%%%%%%%%%%
%\begin{dgroup}
%\begin{align}
%\Phi(b_{=0})=\sum_{s:\alpha_s=0}\sum_{c\in\mathcal{C}^s}\sum_{k\in\mathcal{K}^s_{c}}\sum_{r\in\mathcal{R}_{c}}b^s_{ckr}\log\left(\frac{b^s_{ckr}}{p_{cr}d^s_{ckr}}\right)-\sum_{c\in\mathcal{C}^s}\sum_{k\in\mathcal{K}^s_{c}}b^s_{ck}\log\left(\frac{b^s_{ck}}{w^s_{ck}}\right)
%\end{align}
\begin{dgroup}
\begin{dmath}
\Phi(b_{=1})=\sum_{s:\alpha_s=1}\sum_{c\in\mathcal{C}^s}\sum_{k\in\mathcal{K}^s_{c}}\sum_{r\in\mathcal{R}_{c}}w^s_{ck}b^s_{ckr}\log\left(\frac{b^s_{ckr}}{p_{cr}d^s_{ckr}}\right)
\end{dmath}
\begin{dmath}
\Phi(b_{>1\neq\{\infty\}})=\sum_{s:1<\alpha_s<\infty}\sum_{c\in\mathcal{C}^s}\sum_{k\in\mathcal{K}^s_{c}}\sum_{r\in\mathcal{R}_{c}}b^s_{ckr}\log\left(\frac{b^s_{ckr}}{p_{cr}d^s_{ckr}}\right)-\frac{1}{(1-\alpha_s)}\sum_{c\in\mathcal{C}^s}\sum_{k\in\mathcal{K}^s_{c}}b^s_{ck}\log\left(\frac{b^s_{ck}}{w^s_{ck}}\right)
\end{dmath}
\begin{dmath}
\Phi(b_{=\infty})=\sum_{s:\alpha_s=\infty}\sum_{c\in\mathcal{C}^s}\sum_{k\in\mathcal{K}^s_{c}}\sum_{r\in\mathcal{R}_{c}}b^s_{ckr}\log\left(\frac{b^s_{ckr}}{w^s_{ck}p_{cr}d^s_{ckr}}\right)
\end{dmath}
\end{dgroup}
%\end{align}
As all cases provided above are disjoint, combining them all, we write the complete potential function as 
\begin{align}
\Phi(b)=\Phi(b_{=1})+\Phi(b_{>1\neq\{\infty\}})+\Phi(b_{=\infty})\label{potential}
\end{align}

In the following theorem, we establish the relationship between the potential function $\Phi$ and its dual program $\Upsilon$
\begin{theorem}
Let $\mathbf{b}$ be the spending of service providers, and $x(\mathbf{b})$ be the corresponding allocation according to the TP-mechanism, where $x^s_{ckr}=\frac{b^s_{ckr}}{p_{cr}}$ and $p_{cr}(\mathbf{b})=\sum_{s}\sum_{k}b^s_{ckr}$ then we have following result 
$$\Upsilon(p(b))-\Upsilon(p(b^{*}))\geq \Phi(b^{*})-\Phi(b)$$
where $b^*$ denotes the ME of the market $\mathcal{M}$ 
\end{theorem}
\begin{proof}
The detailed proof is provided in Appendix \ref{constructionpot}
\end{proof}
Moving ahead, we describe some properties of the function $\Phi$, which we will need afterwards to prove the convergence of the proposed bid updating scheme. First, we introduce the definition of the L-Bregman convex function and show that our designed potential function $\Phi$ admits this property.

\begin{definition}[\cite{FisherDynamics}]
 The function $f$ is L-Bregman convex  w.r.t Bregman divergence $d_{g}$ if, for any $y \in rint(C)$ and $x\in C$,
%\begin{dgroup}
\begin{dmath}
f(y)+\left \langle\nabla f(y),x-y  \right \rangle\leq f(x)\leq f(y)+\left \langle\nabla f(y),x-y  \right \rangle +L.d_{g}(x,y)
\end{dmath}
%\end{dgroup}
\end{definition}
In the following lemma, we show that the potential  function $\Phi$ is 1-convex 
depending on the parameter $\alpha_{s}$  employed by the service providers for the fairness criteria.
\begin{lemma}\label{1convex}
The potential function $\Phi(b)$ is  1-Bregman convex w.r.t Bregman divergence $d_{g}$
\begin{align}
d_{g}=\sum_{s:1\leq \alpha_s\leq \infty}KL_a\left(.||.\right)-\sum_{s:1<\alpha_s <\infty}\frac{1}{(1-\alpha_s)}KL_b\left(.||.\right)\label{dg}
\end{align}

where 
\begin{dgroup}
\begin{dmath}
KL_a\left(x||y\right)=\sum_{c}\sum_{k}\sum_{r}x^s_{ckr}\log\left(\frac{x^s_{ckr}}{y^{s}_{ckr}}\right)
\end{dmath}
\begin{dmath}
KL_b\left(x||y\right)=\sum_{c}\sum_{k}x^s_{ck}\log\left(\frac{x^s_{ck}}{y^{s}_{ck}}\right)
\end{dmath}
\label{kldivergence}
\end{dgroup}\label{lema}
where $KL_a$ and $KL_b$ are the Kullback–Leibler divergences
\end{lemma}
\begin{proof}
    Appendix \ref{proofoflema}
\end{proof}

\subsection{Bid updating rule}\label{bidupdating}
In this section, we provide with bid updating rule which enables the service providers to reach the ME of the market $\mathcal{M}$. We build an analysis of the bid updating rule depending on the SPs' $\alpha_s$ fairness criteria. We consider a case when all the SPs employ the $\alpha_s \geq 1$, calling it a complementary domain; over this domain, as we have shown in the previous section, the potential function $\Phi$ is convex in its argument, and the minimal point of $\Phi$ represents the ME points. As the function, $\Phi$ is separable in each SP decision. The SPs can reach the equilibrium by employing a mirror descent update to minimize the potential function $\Phi$  in a decentralized fashion.  Let $b^s_{ckr}(t)$ represent the bid submitted by SP $s$ at step $t$ on the resource type $r$ in the cell $c$ for the class of user $k$, $p_{cr}(t)$ defines the price of the resource set through TP-mechanism in time step $t$, where $p_{cr}(t)=\sum_{s}\sum_{k\mathcal{K}^s_c}b^s_{ckr}(t)$\\
 The bid update for service providers in the time step $t+1$ is given as
\begin{itemize}

\item[•]  if $\alpha_s=\infty$
\begin{align}
b^s_{ckr}(t+1)=
\frac{B_sw^s_{ck}p_{cr}(t)d^s_{ckr}}{\sum_{c}\sum_{k}\sum_{r} w^s_{ck}p_{cr}(t)d^s_{ckr}}\;\label{updateinft}
\end{align}
\item[•]  if  $1\leq \alpha_s <\infty$ 

\begin{align}
\frac{B_{s}\frac{p_{cr}(t)d^s_{ckr}}{\sum_{r}p_{cr}(t)d^s_{ckr}}\left( w^s_{ck}\right)^{\frac{1}{^{\alpha_s}}}\left(\sum_{r}p_{cr}(t)d^s_{ckr}\right)^{\frac{(1-\alpha_s)}{^{-\alpha_s}}}}{\sum_{c}\sum_{k}\left( w^s_{ck}\right)^{\frac{1}{^{\alpha_s}}}\left(\sum_{r}p_{cr}(t)d^s_{ckr}\right)^{\frac{(1-\alpha_s)}{^{-\alpha_s}}}} \label{updategeq1} 
\end{align}
\end{itemize}

Following theoretical results show that if the SPs update their bids according to the above-designed rule, then iterative bid updating dynamics of SPs converge to the market  equilibrium of market $\mathcal{M}$
%\begin{theorem}\label{thmconvg}[Theorem 3.2\cite{FisherDynamics}]
%Suppose f is L-Bregman convex function w.r.t d and $x^{T}$ is the point reached after T applications of mirror descent update rule then
%$f(x^T)-f(x^*)\leq \frac{L d(x^*,x^0)}{T}$ 
%\end{theorem}

\begin{theorem}
Consider each SP $s\in\mathcal{S}$ implement the $\alpha_{s}$-fairness with its respective fairness parameter $\alpha_{s}\in\left[1,\infty \right]$ and repeatedly update their distribution of shares using rule \eqref{updateinft}-\eqref{updategeq1}. Then the potential function $\Phi$ from \eqref{potential} converges to the ME as follows
\begin{dmath}
\Phi(b^{T})-\Phi(b^{*})\leq\frac{1}{T}\sum_{s}KL_a\left(b^*_{s}||b^0_s\right)-\frac{1}{T}\sum_{s:1<\alpha_s<\infty}\frac{1}{(1-\alpha_s)}KL_b\left(b^*_{s}||b^0_s\right)\label{thmresult}
\end{dmath}
where $KL_a$ and $KL_b$ are as defined in \eqref{kldivergence}
\end{theorem}
\begin{proof}
Steps of proof are as follows 
\begin{enumerate}
\item We show that if $\forall s\in\mathcal{S}$  with $\alpha_{s}\geq 1$ update rule \eqref{updateinft} and \eqref{updategeq1} is mirror descent update of potential $\Phi$ w.r.t. Bregman divergence $d_{g}$ \eqref{dg} 

the detailed derivation of update rule is provided in the Appendix \ref{updaterulenegative}, where mirror descent update in step $(t+1)$ is given as

\begin{dmath}
b^{s}(t+1)\longleftarrow\argminA_{\sum_{c}\sum_{k}\sum_{r}{b^s_{ckr}}\leq B_{s}}\left\lbrace \nabla_{b^s} \Phi\left(b^{s}(t)\right)\left(b^s-b^{s}(t)\right)+KL_a\left(b^s||b^{s}(t)\right)-\frac{KL_b\left(b^s||b^{s}(t)\right)}{(1-\alpha_s)})\right\rbrace
\end{dmath}

\item From Lemma \ref{1convex}, we know that $\Phi(b)$ is  $1$-Bergman convex function w.r.t. to $d_{g}$\eqref{dg}.  
\item Now, suppose the $b^{T}$ is the point reached after $T$ applications of the mirror descent update rule then by applying Thm 3.2 \cite{FisherDynamics}, 
we get the desired result \eqref{thmresult}  
\end{enumerate}
\end{proof}

\begin{corollary}
Consider SPs apply the fairness criteria with $\alpha_{s}\geq 1$ and the given price set by the TP-mechanism in each time step they update their distribution of shares in each next time step as the best response. Then the iterative best response dynamics of SPs converges to ME.  
\end{corollary}
\begin{proof}
We show that for all SPs with $\alpha_s \geq 1$, given resources prices announced by the TP-mechanism, the bid update rule in the next round is exactly the best response of SPs given resources prices set in the current round. Hence the convergence of best response dynamics follows from the previous theorem. The derivation of the best response of SP given prices is given in Appendix \ref{bestresponse} 
\end{proof}

\section{The fairness and efficiency }\label{fairnessefficiency}

This section investigates the fairness and efficiency properties of the proposed allocation scheme. We measure the performance of the proposed scheme with the help of the social welfare function; it is a real-valued function that measures the desirability of the allocation outcome. The higher a value it assigns to the outcome, the more desirable the outcomes for a social planner. Various social welfare functions have been mentioned in the literature, the most commonly studied  among them  are the max-min welfare $\Phi({x})=\underset{s}{\text{min}}\;U_s(x_s)$ the Nash welfare $\Phi({x})=\Pi_s U_{s}(x_s)^{B_s}$ utilitarian welfare $\Phi({x})= \sum_s U_{s}(x_s)$. As per the result established in the Section \ref{CentralizedResAlloc}, the market equilibrium for market $\mathcal{M}$ can be computed by solving EG-optimization program \eqref{galeconvex}, Eisenberg and Gale showed in their celebrated work  \cite{eisenberg1959consensus}\cite{eisenberg1961aggregation} that allocation under market equilibria achieves optimal Nash welfare. This result has been established based on a relation that the maximization of the objective function in \eqref{galeconvex} is equivalent to the maximization of Nash welfare function.
\begin{equation}
\argmaxA_{x\in\mathcal{X}}\Pi_s U_{s}(x_s)^{B_s}=\argmaxA_{x\in\mathcal{X}}\sum_{s\in\mathcal{S}}B_s\log\left(U_s(x_s)\right)
\end{equation}
Therefore proposed allocation scheme maximizes the Nash welfare or achieves the proportional fair criteria while distributing the resources among the service providers. 
\subsection{Baseline resource allocations}\label{baseline}
This section presents the two baseline allocation schemes to conduct a comparative analysis of the efficiency of our proposed resource allocation scheme. As discussed earlier, one of the goals of our proposed allocation is to achieve a tradeoff between efficiency versus SLA protection, and we know that the optimal social allocation provides better service utilization. In contrast, static proportional sharing allocation (SS) offers complete protection of SLA among SPs. Thus we consider the socially optimal
allocation and the static proportional sharing scheme as baseline allocation schemes. 
  \\
\textbf{Socially Optimal Allocation (SO):}
In this work, we consider that the utility of each SP is its private information and not known to others. However, If the SPs' utilities were known to the InP, the natural choice of allocation scheme InP could have applied is the socially optimal resource allocation scheme. Thus to compare the efficiency of the proposed allocation scheme, we consider the following social welfare optimization problem.
\begin{subequations}
\begin{align}
&\underset{x}{\text{maximize}}
&&\sum_{s\in\mathcal{S}}B_s \left(U_s(x_s)\right)\\
& \text{subject to}
&&\sum_{s\in\mathcal{S}}\sum_{k\in\mathcal{K}_s}x^s_{ckr}\leq 1, \forall c\in\mathcal{C},r\in\mathcal{R}\\
&&& x^s_{ckr}\geq 0,\forall c\in\mathcal{C},r\in\mathcal{R}
\end{align}\label{fairopti}
\end{subequations}
%The function $f$ is considered to belong to the celebrated class of
%fairness \cite{alpha_fair}, which attains $\alpha$-fairness among the agents while distributing the resources.
%\begin{equation}\label{eq:fairness}
%f(y)=
%\begin{cases}
%\frac{(y)^{1-\alpha}}{(1-\alpha)} & \text{ if } \alpha\neq1 \\ 
% \log(y)& \text{ if } \alpha=1 
%\end{cases}
%\end{equation}
%\textbf{Nash welfare allocation scheme :} In this scheme, the resources are allocated to attain the Nash welfare or fair proportional allocation among the MVNOs.    
%\begin{equation}
%\begin{aligned}
%\;&\underset{d}{\text{maximize}}
%&&\sum_{i\in\mathcal{N}}B_i\log\left(U^i(d^i)\right)\\
%& \text{subject to}
%&&\sum_{i\in\mathcal{N}}d^i_{cr}\leq 1, \forall c\in\mathcal{C},r\in\mathcal{R}\\
%&&& d^i_{cr}\geq 0,\forall c\in\mathcal{C},r\in\mathcal{R}
%\end{aligned}\label{fairopti}
%\end{equation}
\textbf{Static Proportional Sharing (SS)}: It is also known as static proportional splitting. In this resource allocation scheme, resources are partitioned based on the network shares (budgets) of SPs. To be more precise, every SP is allocated a portion of every demanded resource proportional to its budget or shares \emph{i.e.,} $\forall s\in\mathcal{S}, \forall c\in\mathcal{C}$ and $\forall r\in, \mathcal{R}$ $x^s_{cr}=\frac{B_s}{\sum_{s'\in\mathcal{S}}B_{s'}}$. 
\par Now we analyze the efficiency of the proposed scheme \emph{i.e.,} efficiency of ME to the market $\mathcal{M}$ by comparing it with socially optimal allocation. Let $U(SO)$ denotes social optimum, an optimal value of optimization problem defined in \eqref{fairopti}, and $U(ME)$ denotes the value of social welfare under allocation imposed at ME. We consider standard notion of \emph{price of anarchy} defined as $\text{PoA}=\frac{U(SO)-U(ME)}{U(SO)}$. To find the PoA of given market $\mathcal{M}$, we first use the result discussed in the beginning of section that the resource allocation under ME of the market $\mathcal{M}$ can equivalently be computed by solving EG-program and the problem's optimal solution provides the allocation that attains proportional fairness (PF) among the agents. Generally, the attainment of fairness in allocation results in decline in the system's efficiency. The trade-off between efficiency and fairness were well studied in the \cite{Price-of-Fairness} using the notion \emph{price of fairness} (PoF), which is defined as a relative reduction in social welfare under fair allocation compared to the social optimum, $\text{PoF}=\frac{U(SO)-U(PF)}{U(SO)}$. Where $U(SO)$ denoted the value of optimal social welfare function while $U(PF)$ denotes the value of social welfare function at proportional faired allocation. In the following theorem, we derive the bound on PoA for the proposed allocation scheme using results on the bound of  PoF established in \cite{Price-of-Fairness}.
\begin{theorem}
Let the maximum achievable utility of each SP $s\in \mathcal{S}$ in market $\mathcal{M}$ is $\hat{U}_{s}>0$, then price of anarchy is bounded by $\text{PoA}\leq 1- \frac{2\sqrt{S}-1}{S}\frac{\min_{s\in\mathcal{S}}\hat{U}_s}{\max_{s\in\mathcal{S}}\hat{U}_s}-\frac{1}{S}+\frac{\min_{s\in\mathcal{S}}\hat{U}_s}{\sum_{s\in\mathcal{S}}\hat{U}_s}$ and if the maximum achievable utilities of all SPs are equal then $\text{PoA}\leq 1-\frac{2\sqrt{S}-1}{S}$
\end{theorem}
\begin{proof}
Let the price of fairness (PoF) for proportional fairness criteria is defined as 
$\text{PoF}=\frac{U(SO)-U(PF)}{U(SO)}$. According to theorem 2 \cite{Price-of-Fairness} if maximum achievable utility of each agent $s\in \mathcal{S}$ in the market $\mathcal{M}$ is $\hat{U}_{s}>0$ then value of PoF for proportional fairness is bounded by 
\begin{dmath}
    \frac{U(SO)-U(PF)}{U(SO)}\leq 1-\frac{2\sqrt{S}-1}{S}\frac{\min_{s\in\mathcal{S}}\hat{U}_s}{\max_{s\in\mathcal{S}}\hat{U}_s}-\frac{1}{S}+\frac{\min_{s\in\mathcal{S}}\hat{U}_s}{\sum_{s\in\mathcal{S}}\hat{U}_s}
\end{dmath}
We notice that the value of social welfare under proportional fair allocation is equal to the value of social welfare under ME's allocation. Thus replacing $U(PF)$ by $U(ME)$ we get  
\begin{dmath}
    \frac{U(SO)-U(ME)}{U(SO)}\leq 1-\frac{2\sqrt{S}-1}{S}\frac{\min_{s\in\mathcal{S}}\hat{U}_s}{\max_{s\in\mathcal{S}}\hat{U}_s}-\frac{1}{S}+\frac{\min_{s\in\mathcal{S}}\hat{U}_s}{\sum_{s\in\mathcal{S}}\hat{U}_s}    
\end{dmath}

Hence bound on PoA write as 
$$\text{PoA}\leq 1- \frac{2\sqrt{S}-1}{S}\frac{\min_{s\in\mathcal{S}}\hat{U}_s}{\max_{s\in\mathcal{S}}\hat{U}_s}-\frac{1}{S}+\frac{\min_{s\in\mathcal{S}}\hat{U}_s}{\sum_{s\in\mathcal{S}}\hat{U}_s}$$
Similarly if the maximum achievable utilities of all SPs are equal then by the Thm.1 \cite{Price-of-Fairness} 
$$\text{PoA}\leq 1-\frac{2\sqrt{S}-1}{S}$$ 
\end{proof} From the above theorem, we can deduce that the efficiency of the proposed resource allocation scheme decreases with the increase in the number of SPs. Nonetheless, the social optimal allocation offers efficient resource utilization, but at the cost of poor fairness. In the numerical section, we will see that sometimes hardly any resources are allocated for SPs with low marginal gain under the SO allocation scheme. Further, the SO allocation scheme does not guarantee the existence of any equilibrium or stability in the allocation method. Next, we compare the performance of the proposed scheme with the static proportional allocation scheme  
\begin{theorem}
 Under the proposed resource allocation scheme, \emph{i.e} at the ME of the market $\mathcal{M}$ each SP achieves the utility higher than or equal to the utility under static proportional allocation (SS). \label{gainoverstatic}
\end{theorem}
\begin{proof}
Let $\left(\hat{\mathbf{x}},\hat{\mathbf{p}}\right)$ be the market equilibrium of market $\mathcal{M}$ then by definition of ME 

$$\forall s\in\mathcal{S}\quad\sum_{ c \in {\cal C}}\sum_{k\in\mathcal{K}^s_c}\sum_{r\in\mathcal{R}_{c}} \hat{p}_{cr} \hat{x}^s_{ckr}=B_{s}$$
and $\forall s\in\mathcal{S}, U_{s}(\hat{x}_s) $ is utility achived by each SP $s$ under ME.  
Let $\overline{\mathbf{x}}$ be the resource allocation under the static proportional allocation scheme. Then to prove the desired results, we first show that allocation $\overline{\mathbf{x}}$ is budget exhausting with respect to price vector $\hat{\mathbf{p}}$ \emph{i.e,}
$$\text{}\forall s\in\mathcal{S}\quad\sum_{ c \in {\cal C}}\sum_{r\in\mathcal{R}_{c}} \hat{p}_{cr} \overline{x}^s_{cr}\leq B_s$$
Just by replacing $\overline{x}^s_{cr}$ with $\frac{B_s}{\sum_{s'\in\mathcal{S}}B_{s'}}$ in above inequality gives us first result. Now as $\hat{\mathbf{x}}$ and $\overline{\mathbf{x}}$ both are feasible and budget exhausting allocation then by definition of market equilibrium 
$$U_{s}\left(\overline{x}_{s}\right)\leq U_{s}(\hat{x}_s), \forall s\in\mathcal{S}\ $$
Hence proves the theorem 
\end{proof}
The above theoretical result proves that the proposed resource scheme achieves better efficiency than the static proportional sharing scheme. Thus proposed FM-based allocation bring off a better arbitrage between the system efficiency and protection to the service level agreement of SPs
\section{Numerical Experiments}
\label{numerical}

%%%%%%%%%%%%%%%%%%%%%%%%%%%%%%%%%%%

%CSE_SP3
 In this section, we numerically evaluate the performance of the proposed allocation scheme. For the simulation purpose, we consider a scenario where the network consists of seven cells, and each cell accommodates three types of resources: CPU, RAM, and Bandwidth (BW), and their available capacities at each cell are 30 Units 126Gb and 40MHz, respectively. Initially, we assume the three service providers owning an equal share of infrastructure support the four types of service classes: CPU-intensive, RAM-intensive, BW-intensive, and Balanced class. The base demand vector for each service class is as described in Table \ref{basedemandtable} \cite{Moro2020JointMO}.
\begin{table}[h]

\centering
\begin{tabular}{|c|c|c|c|}
\hline 
Service Class & CPU & RAM & BW \\ 
\hline 
BW-Intensive & 1CPUs & 8Gb & 10MHz \\ 
\hline 
CPU-Intensive & 4 CPU & 8Gb & 3MHz \\ 
\hline 
RAM-Intensive & 1 CPU & 32Gb & 3MHz \\ 
\hline 
Balanced & 5CPUs & 40Gb & 5MHz \\ 
\hline 
\end{tabular} 
\caption{The base demand vector of service classes}\label{basedemandtable}
\end{table}
For convenience, we assume that each SP supports two user classes. Details of the user class supported by each SP are given in Table \ref{distributionofusers}. For each user class at each cell, a user load is normally distributed with a mean of $100$ and a variance of $50$. We create random $2000$ instances with the above distribution.
\begin{table}[h]

\centering
\begin{tabular}{|c|c|}
\hline 
Service Provider & Classes supported  \\ 
\hline 
SP1 & BW-Intensive, Balanced \\ 
\hline 
SP2 & CPU-Intensive, Balanced\\ 
\hline 
SP3 & RAM-Intensive, Balanced\\ 
\hline 
\end{tabular}
\caption{Class of users supported by service providers }\label{distributionofusers}
\end{table}
\par For each of the generated instances, assuming the SPs employ the different alpha fairness criteria with alpha in the range [0,5], we use SO, ME, and SS-based resource allocation schemes. We calculate the market equilibrium allocations for each generated scenario for our proposed resource allocation scheme, employing the designed bid updating procedure in Section \ref{bidupdating}. We investigate the performance of the proposed allocation scheme by comparing the effect of the alpha-fairness parameter, social welfare, and service providers' network share on the average service rate or utility achieved by their users at different locations.
%\begin{figure*}[h!]
%\includegraphics[width=\textwidth]{./fig_latex/bargraphalpha0110.eps}
%\caption{Bar graphs describe the utilities (service rate) under Socially optimal allocation (SO), market equilibrium (ME) and Static proportional (SS) allocation scheme attained by the users associated with the service providers  SP1, SP2, SP3 and SP4 in the cell 1-4, when $\alpha$ fairness criteria with  $\alpha=0$, $\alpha=1$, $\alpha=2$ and  $\alpha=5$ applied by the service providers.}
%\label{fig:servicerate}
%\end{figure*}

\subsection{$\alpha$-fairness effect}
\begin{figure*}
\includegraphics[width=0.33\textwidth]{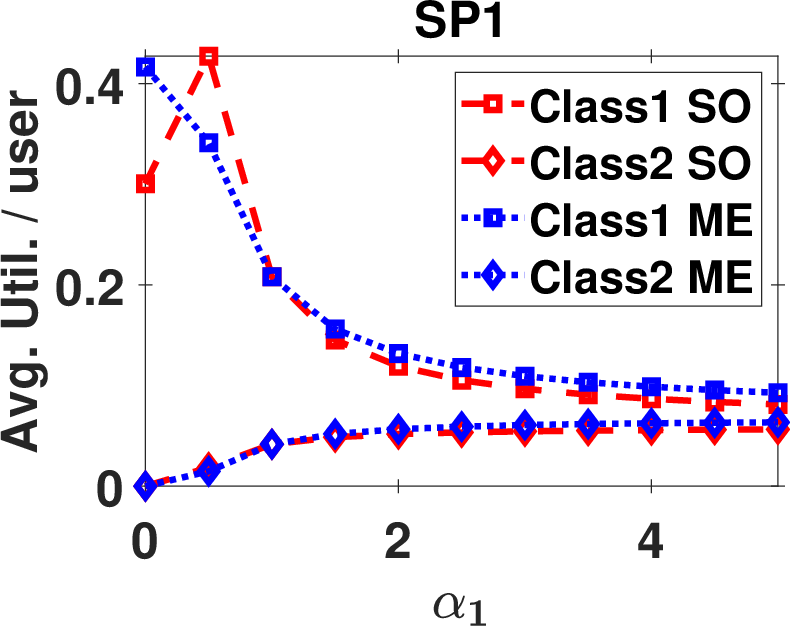}
\includegraphics[width=0.33\textwidth]{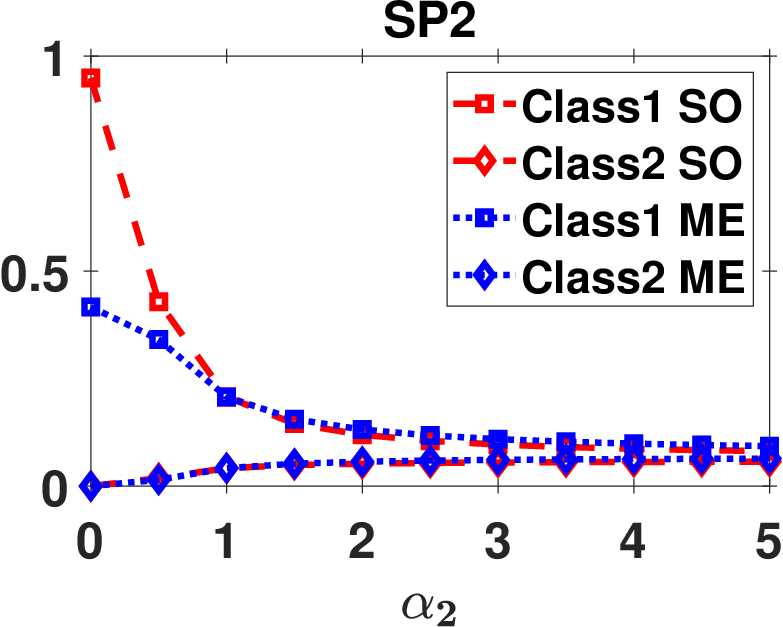}
\includegraphics[width=0.33\textwidth]{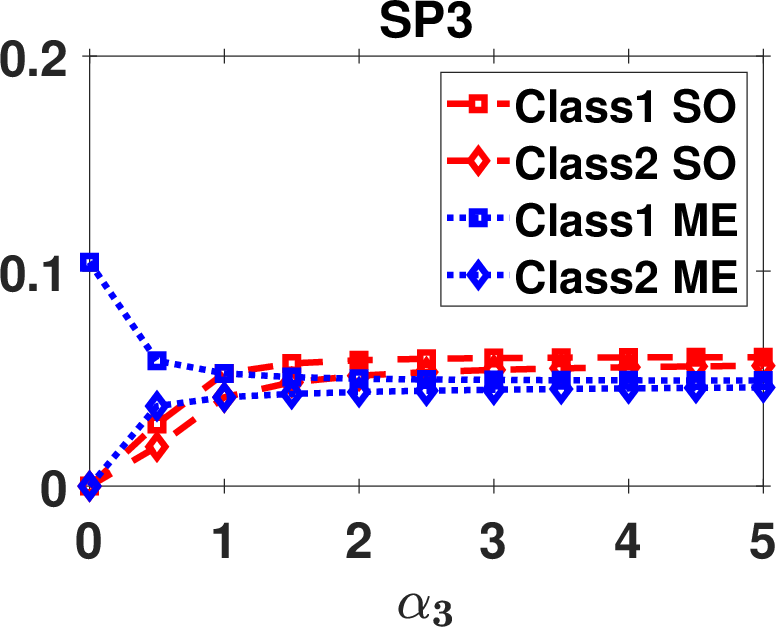}\put(-510,140){(\ref{Alpha_effect}.a)}
\put(-340,140){(\ref{Alpha_effect}.b)}
\put(-170,140){(\ref{Alpha_effect}.c)}
\caption{Plots in Fig. show the effect of alpha fairness criteria ($\alpha_s$) on the average service utility seen by class 1 and class 2 users from Service providers SP1, SP2 and SP3 under SO and ME-based schemes. }

\label{Alpha_effect}
\end{figure*}
Fig.\ref{Alpha_effect} describes the effect of the alpha parameters employed by the service providers on the service rate seen by their different user classes. From plots, we observed that as the value of the alpha parameter increases, the difference between the average service rate experienced by the end-user of two different application classes decreases. More precisely, as the alpha increases, the service rate of the deprived class rises while the service rate of the class that has experienced a high service rate decreases. This phenomenon is coherent with social fairness, where a higher alpha value results in a more equitable distribution of resources, promoting social fairness. However, this is not strictly followed in the case of social optimal allocation. When all the SPs employ the fairness criteria with $\alpha=0$, the SP level fairness and class level fairness among the SPs can not be distinguished separately. As a result, one of the two fairness can dominate the other.

\subsection{Social welfare}
\begin{figure*}
\includegraphics[width=0.33\textwidth]{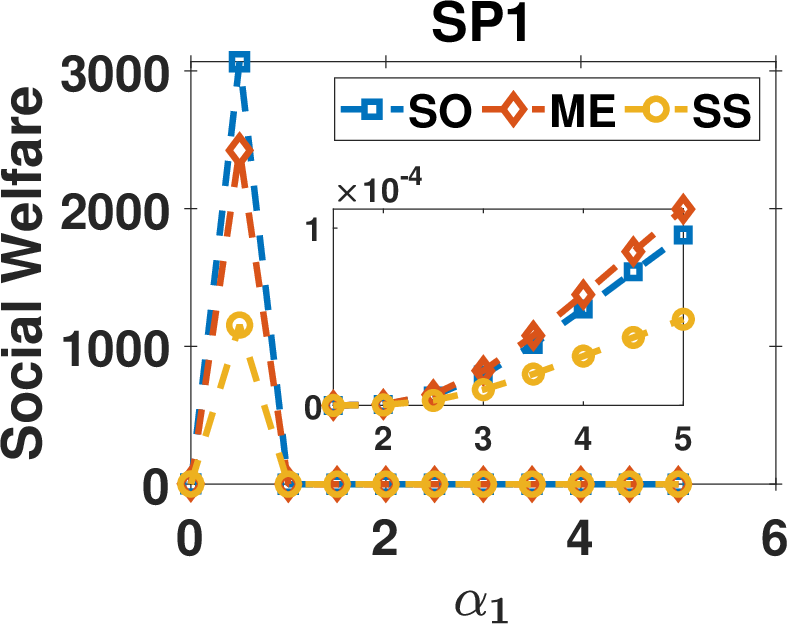}
\includegraphics[width=0.33\textwidth]{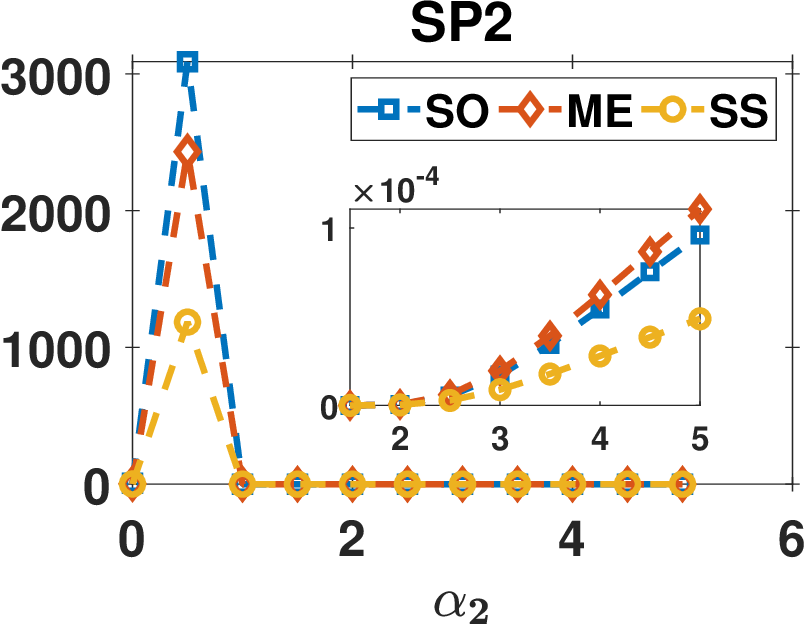}
\includegraphics[width=0.33\textwidth]{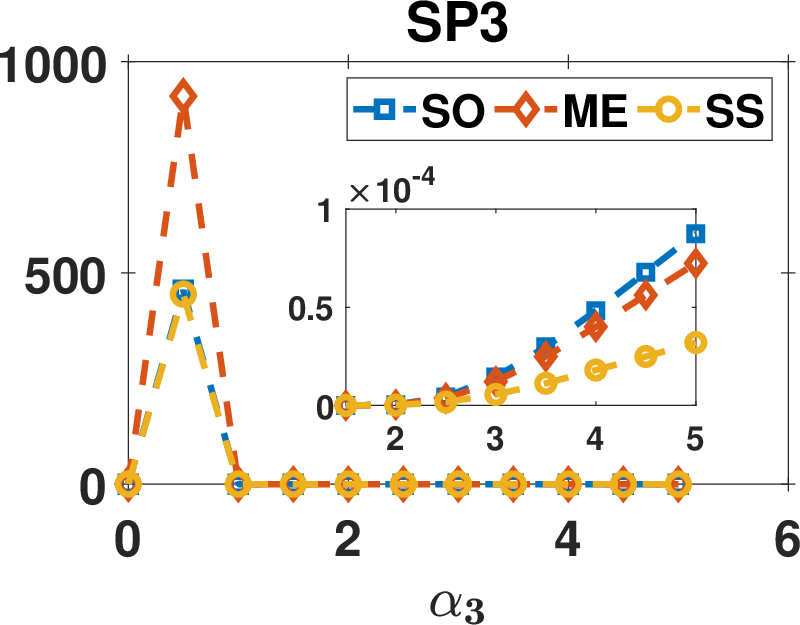}
\put(-510,140){(\ref{soc}.a)}
\put(-340,140){(\ref{soc}.b)}
\put(-170,140){(\ref{soc}.c)}
\caption{Plots show the utilities (social welfare) gained by the SPs under ME-based and SO schemes with values of $\alpha$ parameter ranging from 0-5. }
\label{soc}
\end{figure*}
In Fig.\ref{soc}, we study the performance of our proposed allocation scheme by comparing SPs' utility (i.e. the social welfare) achieved under different alpha fairness criteria through the various baseline resource allocation schemes. We observe that for service providers SP1 and SP2, when they employ the $\alpha$ fairness criteria with $\alpha$ in the range of 0 to 1, the utility achieved through the SO scheme is higher than the utility achieved through the market-based scheme, however,  with $\alpha$ in the range of 1 to 5, the utility achieved through the SO scheme is lower than the utility achieved through the market-based scheme. At the same time, in the case of SP3, the results are exactly opposite. Overall results show that the social welfare achieved in ME-based allocation has less variance than achieved through SO. Moreover, as we have theoretically proven in Theorem~\ref{gainoverstatic}, we observe that social welfare gained by all the SPs through the ME-based scheme is higher than those achieved through the SS scheme.

\subsection{Sensitivity to the budget} 
\begin{figure*}
\includegraphics[width=0.33\textwidth]{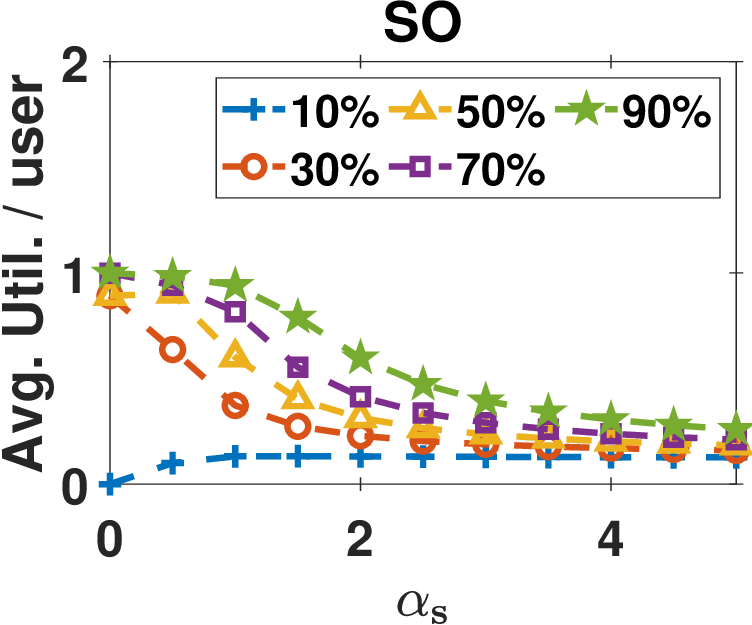}
\includegraphics[width=0.33\textwidth]{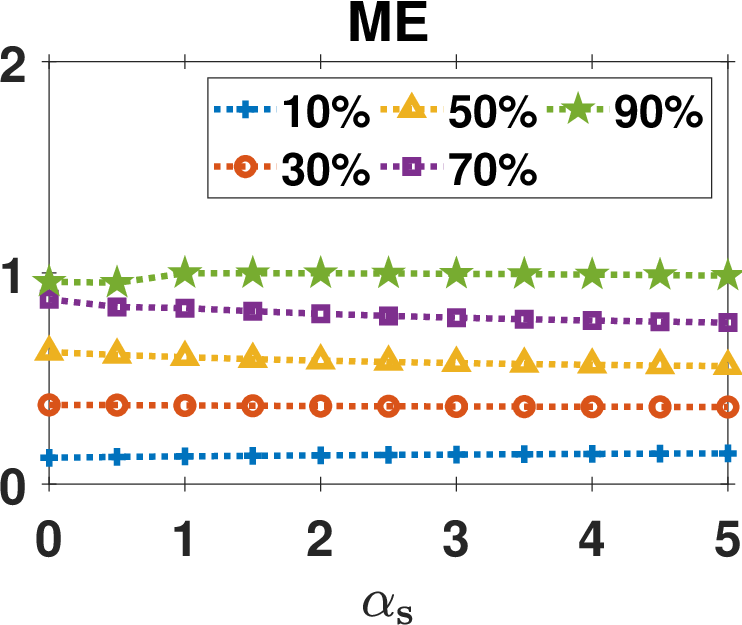}
\includegraphics[width=0.33\textwidth]{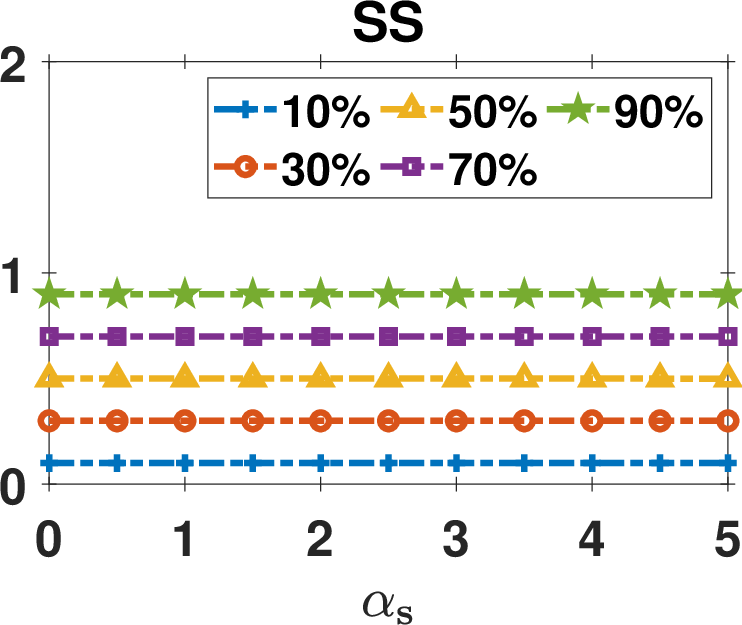}
\put(-510,140){(\ref{sens}.a)}
\put(-340,140){(\ref{sens}.b)}
\put(-170,140){(\ref{sens}.c)}
\caption{Plots describe the sensitivity of the average service rate seen by users of SP1 to the service provider's budget under SO, ME and SS-based resource allocation scheme, where the network share of SP1 varies from $10\%$ to $90\%$ and the respective value of service rate are plotted at different alpha fairness parameter.  }
\label{sens}
\end{figure*}
To study the sensitivity of budgets to different resource allocation schemes, we vary the network share (budget) of a first service provider (SP1) from $10\%$ to $90\%$, while the remaining stake is equally distributed among the reaming service providers. Fig.\ref{sens} describes the effect of the change in the budget on the average service rate seen by the users of service provider SP1. We observe in Fig.(\ref{sens}a) that the socially optimal allocation is very sensitive to the budget. Under the same scheme,   when the SP employs the utilitarian approach, i.e. using the $\alpha_1=0$ and the network share of SP1 is low, end users experience a very low average service rate. However, the sensitivity decreases as the SP employs the higher alpha fairness criteria. From Fig.(\ref{sens}.a), we observe that the Market-based scheme is all most uniformly sensitive to budget change over all the range of alpha parameters. The change in the service rate is linear with the change in the budget. 

\par Finally, we consider the case when SPs employ proportional fairness criteria. The plot in Fig.\ref{fastconvergence} shows the fast convergence of prices to ME from one of the instances at cell $2$ through the designed bid updating scheme.    

\begin{figure}
\includegraphics[width=\columnwidth]{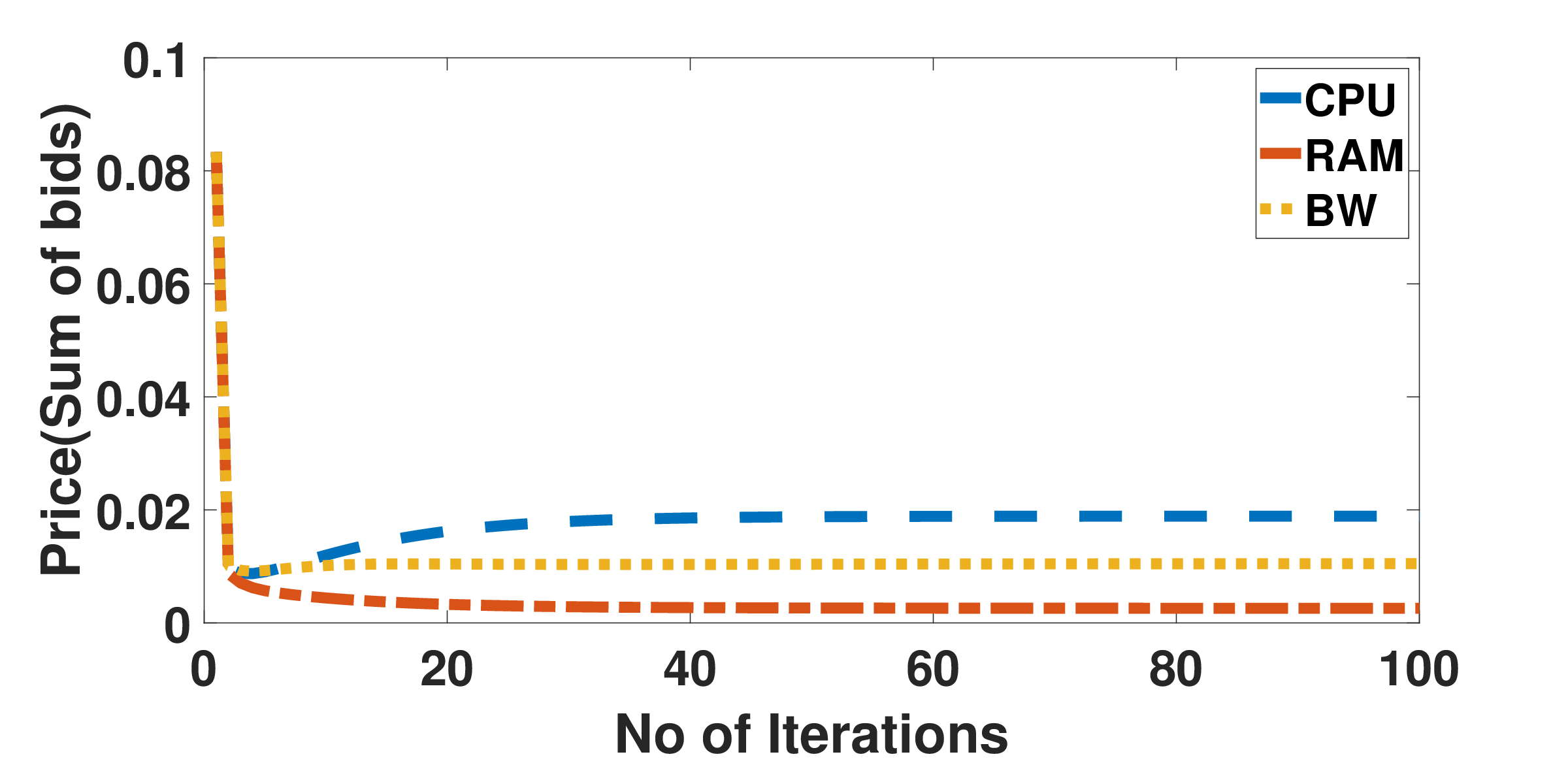}
\caption{The fast convergence of prices to market equilibrium at cell $2$ by SPs via trading post mechanism and when SPs employ proportional fairness criteria.}
\label{fastconvergence}
\end{figure}
%%%%%%%%%%%%%%%%%%%%%%%%%%%%
\section{Conclusion}
In this work, we have proposed a novel flexible multi resource-sharing scheme based on the Fisher market model and the Trading Post mechanism. Our scheme aims to achieve efficient resource utilization while ensuring multi-level fairness and protection of service-level agreements among service providers. We have demonstrated that the proposed allocation scheme operates at market equilibrium, which can be obtained through a simple convex optimization program. Additionally, we have designed a privacy-preserving budget-distributing scheme that enables service providers to converge to market equilibrium in a decentralized manner. The performance of the proposed scheme has been evaluated through theoretical analysis, comparing it with both social optimal and static proportional sharing scheme. The findings clearly highlight the scheme's ability to effectively balance the utilization of available resources while ensuring adequate protection for individual slices. 

\bibliographystyle{IEEEtranN}
\bibliography{IEEEabrv,slicing}
\appendix
\subsection{Construction of potential function and properties}\label{constructionpot}
To construct the potential function we consider the objective function for dual of optimization program \eqref{galeconvex}
\begin{dmath}
\Upsilon (p)= \max_{ \sum_{c,k,r}x^s_{ckr}p_{cr}=B_s}\left ( \sum_{s\in\mathcal{S}}B_s \log\left ( U_s \right )+\sum_{c\in\mathcal{C}}\sum_{r\in\mathcal{R}_c} p_{cr} \left (1-\sum_{s\in\mathcal{S}}\sum_{k\in\mathcal{K}^s_c}x^s_{crk}\right )\right ) 
\end{dmath}

and in dual goal is to minimize the $\Upsilon (p)$ wrt $p$
%
%\begin{dmath}
%\Upsilon (p)= \max_{\sum_{c,k,r}x^s_{ckr}p_{cr}=B_s}\left ( \sum_{s\in\mathcal{s}}B_s \log\left ( U_s \right )+\sum_{c\in\mathcal{C}}\sum_{r\in\mathcal{R}_c} p_{cr} \left (1-\sum_{s}\sum_{k\in\mathcal{K}^s_c}x^s_{crk}\right )\right ) 
%\end{dmath}

\begin{dmath}
\Upsilon (p)= \max_{\sum_{c,k,r}x^s_{ckr}p_{cr}=B_s}\left ( \sum_{s}\frac{B_s \log}{(1-\alpha_s)}\times \left(\sum_{c\in\mathcal{C}^s}\sum_{k\in\mathcal{K}^s_{c}} w^s_{ck}\left( \min_{r}\lbrace\frac{x^s_{ckr}}{d^s_{ckr}}\rbrace\right)^{(1-\alpha_s)}\right)\right)\\
+\sum_{c\in\mathcal{C}}\sum_{r\in\mathcal{R}_c} p_{cr} \left (1-\sum_{s}\sum_{k\in\mathcal{K}^s_c}x^s_{crk}\right )
\end{dmath}

\begin{dmath}
\Upsilon (p)= \max_{\sum_{c,k,r}b^s_{ckr}=B_s}\left ( \sum_{s}\frac{B_s}{(1-\alpha_s)}\times\log\left(\sum_{c\in\mathcal{C}^s}\sum_{k\in\mathcal{K}^s_{c}} w^s_{ck}\left( \min_{r}\lbrace\frac{b^s_{ckr}}{p_{cr}d^s_{ckr}}\rbrace\right)^{(1-\alpha_s)}\right)+\sum_{c\in\mathcal{C}}\sum_{r\in\mathcal{R}_c} p_{cr}-\sum_{s}B_{s}\right )\label{append1} 
\end{dmath}

Let given price vector $p$, $b(p)$ be the spending that maximizes the objective in \eqref{append1} subject to constraints $\left(\sum_{c}\sum_{k}b^s_{ckr}=B_{s}\right)\forall s \in\mathcal{S}$ 
\begin{dmath}
\Upsilon (p)=\left ( \sum_{s}\frac{B_s}{(1-\alpha_s)}\times\log\left(\sum_{c\in\mathcal{C}^s}\sum_{k\in\mathcal{K}^s_{c}} w^s_{ck}\left( \min_{r}\lbrace\frac{b^s_{ckr}(p)}{p_{cr}d^s_{ckr}}\rbrace\right)^{(1-\alpha_s)}\right)+\sum_{c\in\mathcal{C}}\sum_{r\in\mathcal{R}_c} p_{cr}-\sum_{s}B_{s}\right ) 
\end{dmath}
We observe that the at optimal point 
\begin{enumerate}
\item for each $s$, $\forall c$ and  $\forall kw^s_{ck}\frac{{u^s_{ck}}^{(1-\alpha_s)}}{b^s_{ck}}=w^s_{ck}\frac{{u^s_{ck'}}^{(1-\alpha_s)}}{b^s_{ck'}}$ 
\item $\frac{b^s_{ckr}(p)}{p_{cr}d^s_{ckr}}=\frac{b^s_{ckr}(p)}{p_{cr}d^s_{ckr}}$
\end{enumerate}
\begin{dmath}
\Upsilon (\mathbf{p})=\sum_{s\in\mathcal{S}}\sum_{c\in\mathcal{C}^s}\sum_{k\in\mathcal{K}^s_{c}}\sum_{r\in\mathcal{R}_{c}}b^s_{ckr}(\mathbf{p})\log\left(\frac{b^s_{ckr}(\mathbf{p})}{p_{cr}d^s_{ckr}}\right)-\frac{1}{(1-\alpha_s)}\sum_{s\in\mathcal{S}}\sum_{c\in\mathcal{C}^s}\sum_{k\in\mathcal{K}^s_{c}}b^s_{ck}(\mathbf{p})\log\left(\frac{b^s_{ck}(\mathbf{p})}{w^s_{ck}} \right)\label{append2}
\end{dmath}
Given $\mathbf{p}$, the function in right hand side of equation \eqref{append2} is convex in $b$ and its minimal point subject to constraints $\left(\sum_{c}\sum_{k}b^s_{ckr}=B_{s}\right)\forall s \in\mathcal{S}$ is $b(\mathbf{p})$

\begin{dmath}
\Upsilon (\mathbf{p})\leq\sum_{s\in\mathcal{S}}\sum_{c\in\mathcal{C}^s}\sum_{k\in\mathcal{K}^s_{c}}\sum_{r\in\mathcal{R}_{c}}b^s_{ckr}\log\left(\frac{b^s_{ckr}}{p_{cr}d^s_{ckr}}\right)-\frac{1}{(1-\alpha_s)}\sum_{s\in\mathcal{S}}\sum_{c\in\mathcal{C}^s}\sum_{k\in\mathcal{K}^s_{c}}b^s_{ck}\log\left(\frac{b^s_{ck}}{w^s_{ck}} \right)
\end{dmath}
The inequality becomes equality if $b_{s}=b_{s}(\mathbf{p})$
%%%%%%%%%%%%%%%%%%%%%%%
\begin{dmath}
\Upsilon (\mathbf{p}(b))\leq \Phi(b)\label{append3}
\end{dmath}
Since we know that the $b^*=b\left(p(b^*)\right)$
\begin{dmath}
\Upsilon (\mathbf{p}(b^*))=\Phi(b^*)\label{append4}
\end{dmath}
%%%%%%%%%%%%%%%%%%%%%%%%%%%%%
From \eqref{append3} and \eqref{append4} we have 
\begin{equation}
\Upsilon(p(b))-\Upsilon(p(b^{*}))\leq \Phi(b)-\Phi(b^{*})
\end{equation}
\subsection{Proof of Lemma \ref{lema}}\label{proofoflema}
\begin{proof}
  We know that the $\Phi(b)$ is a convex function then by definition of a convex function, we have 
\begin{align}
\Phi(b')+\left \langle\nabla \Phi(b'),b-b'  \right \rangle\leq \Phi(b')
\end{align}
Now consider 
\begin{equation}
\Phi(b)-\Phi(b')-\left \langle \nabla \Phi(b'),b-b' \right \rangle
\end{equation}
 putting the value of $\nabla \Phi(b')$ in above equation and after some calculations, we get 
\begin{dmath}
\Phi(b)-\Phi(b')-\left \langle \nabla \Phi(b'),b-b' \right \rangle=\sum_{s:1\leq \alpha_s\leq \infty}KL_{a}(b_s||b'_s)-\sum_{s:1<\alpha_s<\infty}\frac{1}{(1-\alpha_s)}KL_{b}(b^s|{|b^s}')-KL(p||p')
\end{dmath}
and since $KL(p||p')$ is non negative
\begin{dmath}
\Phi(b)-\Phi(b')-\left \langle \nabla \Phi(b'),b-b' \right \rangle\leq\sum_{s:1\leq \alpha_s\leq \infty}KL_{a}(b_s||b'_s)-\sum_{s:1<\alpha_s<\infty}\frac{1}{(1-\alpha_s)}KL_{b}(b_s|{|b_s}')
\end{dmath}
\begin{dmath}
\Phi(b)\leq \Phi(b')+\left \langle \nabla \Phi(b'),b-b' \right \rangle+\leq\sum_{s:1\leq \alpha_s\leq \infty}KL_{a}(b_s||b'_s)-\sum_{s:1<\alpha_s<\infty}\frac{1}{(1-\alpha_s)}KL_{b}(b_s|{|b_s}')
\end{dmath}
which proves that function $\Phi(b)$ is $1$ Bergman convex wrt $d_{g}$ $\eqref{dg}$
\end{proof}
\subsection{Derivation of update rule for $\alpha_s\geq1$}\label{updaterulenegative}
\begin{dmath}
b_{s}(t+1)=\argminA_{\sum_{c}\sum_{k}\sum_{r}{b^s_{ckr}}\leq B_{s}}\left\lbrace \nabla_{b^s} \Phi_{p}\left(b^{s}(t)\right)\left(b^s-b^{s}(t)\right)+KL_a\left(b^s||b^{s}(t)\right)-\frac{1}{(1-\alpha_s)}KL_b\left(b^s||b^{s}(t)\right))\right\rbrace
\end{dmath}
we consider the Lagrangian 
\begin{dmath}
L_s\left(b^s,\gamma\right)=\nabla_{b^s} \Phi_{p}\left(b^{s}(t)\right)\left(b^s-b^{s}(t)\right)+KL_a\left(b^s||b^{s}(t)\right)\\
-\frac{1}{(1-\alpha_s)}KL_b\left(b^s||b^{s}(t)\right))+\gamma\left(\sum_{c}\sum_{k}\sum_{r}{b^s_{ckr}}- B_{s}\right)
\end{dmath}

%\begin{dmath}
%\argminA_{\sum_{c}\sum_{k}\sum_{r}{b^s_{ckr}}\leq B_{s}}\left\lbrace \sum_{c}\sum_{k}\sum_{r}\nabla_{b^s_{ckr}} \Phi_{p}\left(b^{s}_{ikr}(t)\right)\left(b^s_{ckr}-b^{s}_{ckr}(t)\right)+\sum_{k}\sum_{r}b^s_{ckr}\log\left(\frac{b^s_{ckr}}{b^{s}_{ikr}(t)}\right)\right\rbrace
%\end{dmath}
After applying the first order KKT condition we get
\begin{dmath}
\left[1+\log\left(\frac{b^{s}_{ckr}(t)}{p_{cr}d^s_{ckr}}\right)\right]-\frac{1}{(1-\alpha_s)}\left[1+\log\left(\frac{b^s_{ck}(t)}{w^s_{ck}}\right)\right]+1\\
+\log
\left(\frac{b^s_{ckr}}{b^{s}_{ckr}(t)}\right)-\frac{1}{(1-\alpha_s)}\left[1+\log\left(\frac{b^s_{ck}}{b^{s}_{ck}(t)}\right)\right]+\gamma=0
\end{dmath}
After some calculation we have 
%\begin{dmath}
%-\log\left( p_{cr}d^s_{ckr}\right)-\frac{1}{(1-\alpha_s)}\log\left(\frac{b^s_{ck}}{w^s_{ck}}\right)+\log\left(b^s_{ckr}\right)+\gamma=0
%\end{dmath}
\begin{dmath}
\log\left(b^s_{ckr}\right)=\log\left( p_{cr}d^s_{ckr}\right)+\frac{1}{(1-\alpha_s)}\log\left(\frac{b^s_{ck}}{w^s_{ck}}\right)-\gamma
\end{dmath}
taking exponentila on both side 
\begin{dmath}
b^s_{ckr}=e^{C}p_{cr}d^s_{ckr}\left(\frac{b^s_{ck}}{w^s_{ck}}\right)^{1/(1-\alpha_s)}\label{bsckr}
\end{dmath}
summing over $r$ on both side  
\begin{dmath}
b^s_{ck}\left(\frac{b^s_{ck}}{w^s_{ck}}\right)^{-1/(1-\alpha_s)}=e^{C}\sum_{r}p_{cr}d^s_{ckr}
\end{dmath}
%\begin{dmath}
%\left(b^s_{ck}\right)^{\frac{-\alpha_s}{(1-\alpha_s)}}=e^{c}\left( w^s_{ck} \right)^{-1/(1-\alpha_s)}\left(\sum_{r}p_{cr}d^s_{ckr}\right)
%\end{dmath}
\begin{dmath}
\left(b^s_{ck}\right)^{\frac{1}{(1-\alpha_s)}}=e^{c\frac{1}{-\alpha_s}}\left( w^s_{ck} \right)^{\frac{-1}{(-\alpha_s)(1-\alpha_s)}}\left(\sum_{r}p_{cr}d^s_{ckr}\right)^{\frac{1}{-\alpha_s}}\label{bscreplace}
\end{dmath}
replacing $\left(b^s_{ck}\right)^{\frac{1}{(1-\alpha_s)}}$ from \eqref{bscreplace} in \eqref{bsckr} we have 
\begin{dmath}
b^s_{ckr}=e^{c\frac{-\alpha_s}{(1-\alpha_s)}}p_{cr}d^s_{ckr}\left( w^s_{ck} \right)^{\frac{-1}{(-\alpha_s)}}\left(\sum_{r}p_{cr}d^s_{ckr}\right)^{\frac{1}{-\alpha_s}}
\end{dmath}
summing over all $r$ across all the classes and cell we have 
\begin{dmath}
B_{s}=e^{c\frac{-\alpha_s}{(1-\alpha_s)}}\sum_{c}\sum_{k}\left( w^s_{ck} \right)^{\frac{-1}{(-\alpha_s)}}\left(\sum_{r}p_{cr}d^s_{ckr}\right)^{\frac{(1-\alpha_s)}{-\alpha_s}}
\end{dmath}

\begin{dmath}
b^s_{ckr}=\frac{B_{s}p_{cr}d^s_{ckr}\left( w^s_{ck} \right)^{\frac{-1}{(-\alpha_s)}}\left(\sum_{r}p_{cr}d^s_{ckr}\right)^{\frac{1}{-\alpha_s}}}{\sum_{c}\sum_{k}\left( w^s_{ck} \right)^{\frac{-1}{(-\alpha_s)}}\left(\sum_{r}p_{cr}d^s_{ckr}\right)^{\frac{(1-\alpha_s)}{-\alpha_s}}}
\end{dmath}
\begin{dmath}
b^s_{ckr}=\frac{B_{s}\frac{p_{cr}d^s_{ckr}}{\sum_{r}p_{cr}d^s_{ckr}}\left( w^s_{ck}\right)^{\frac{1}{^{\alpha_s}}}\left(\sum_{r}p_{cr}d^s_{ckr}\right)^{\frac{(1-\alpha_s)}{^{-\alpha_s}}}}{\sum_{c'\in\mathcal{C}}\sum_{k'\in\mathcal{K}^s_{c'}}\left( w^s_{c'k'}\right)^{\frac{1}{^{\alpha_s}}}\left(\sum_{r'}p_{c'r'}d^s_{c'k'r'}\right)^{\frac{(1-\alpha_s)}{^{-\alpha_s}}}}
\end{dmath}

\subsection{Best response}\label{bestresponse}
We consider the Lagrangian for the optimization \eqref{marketprob}
\begin{dmath}
L_s\left(u,b,\lambda,\gamma\right)=\frac{1}{(1-\alpha_s)}\sum_{c}\sum_{k} w^s_{ck} \left({u^s_{ck}}\right)^{(1-\alpha_s)}-\sum_{c}\sum_{k}\sum_{r}\lambda_{ckr}\left(u^{s}_{ck}-\frac{b^{s}_{ckr}}{p_{cr}d^s_{ckr}}\right)-\gamma\left(\sum_{c}\sum_{k}\sum_{r}b^{s}_{ckr}-B_s\right)
\end{dmath}
After applying the first order KKT conditions 
\begin{align}
w^s_{ck}\left(u^s_{ck}\right)^{-\alpha_s}-\sum_{r}\lambda_{ckr}=&0,\forall c\in\mathcal{C},\forall k\in\mathcal{K},\forall\label{bestkkt1}\\
\frac{\lambda_{ckr}}{p_{cr}d^s_{ckr}}=&\gamma\;\forall c\in\mathcal{C},\forall k\in\mathcal{K},\forall r\in\mathcal{R}\label{bestkkt2}
\end{align}
from \eqref{bestkkt1} and \eqref{bestkkt2}
\begin{dmath}
w^s_{ck}\left({u^s_{ck}}\right)^{-\alpha_s}=\gamma\sum_{r}p_{cr}d^s_{ckr}\\
\end{dmath}
\begin{dmath}
{u^s_{ck}}=\left(\gamma\right)^{\frac{1}{^{-\alpha_s}}}\left( w^s_{ck}\right)^{\frac{1}{^{\alpha_s}}}\left(\sum_{r}p_{cr}d^s_{ckr}\right)^{\frac{1}{^{-\alpha_s}}}
\end{dmath}
we know that at best response $u^s_{ck}=\frac{b^s_{ckr}}{p_{cr}d^s_{ckr}}$
\begin{dmath}
\frac{b^s_{ckr}}{p_{kr}d^s_{ckr}}=\gamma^{\frac{1}{^{-\alpha_s}}}\left( w^s_{ck}\right)^{\frac{1}{^{\alpha_s}}}\left(\sum_{r}p_{cr}d^s_{ckr}\right)^{\frac{1}{^{-\alpha_s}}}
\end{dmath}
\begin{dmath}
b^s_{ckr}=\gamma^{\frac{1}{^{-\alpha_s}}}\left( w^s_{ck}\right)^{\frac{1}{^{\alpha_s}}}p_{cr}d^s_{ckr}\left(\sum_{r}p_{kr}d^s_{ckr}\right)^{\frac{1}{^{-\alpha_s}}}\label{bsckr4}
\end{dmath}
%\begin{dmath}
%\sum_{r}b^s_{ckr}=\gamma^{\frac{1}{^{-\alpha_s}}}\left( w^s_{ck}\right)^{\frac{1}{^{\alpha_s}}}\sum_{r}p_{cr}d^s_{ckr}\left(\sum_{r}p_{cr}d^s_{ckr}\right)^{\frac{1}{^{-\alpha_s}}}
%\end{dmath}
summing over $c\in\mathcal{C}$, $k\in\mathcal{K}$ and $r\in\mathcal{KR}$
\begin{dmath}
B_{s}=\gamma^{\frac{1}{^{-\alpha_s}}}\sum_{c}\sum_{k}\left( w^s_{ck}\right)^{\frac{1}{^{\alpha_s}}}\left(\sum_{r}p_{cr}d^s_{ckr}\right)^{\frac{(1-\alpha_s)}{^{-\alpha_s}}}\label{bsck4}
\end{dmath}
Replacing the value of $\gamma^{\frac{1}{^{-\alpha_s}}}$ from \eqref{bsck4} in \eqref{bsckr4}
\begin{dmath}
b^s_{ckr}=\frac{B_{s}p_{cr}d^s_{ckr}\left( w^s_{ck}\right)^{\frac{1}{^{\alpha_s}}}\left(\sum_{r}p_{cr}d^s_{ckr}\right)^{\frac{1}{^{-\alpha_s}}}}{\sum_{c}\sum_{k}\left( w^s_{ck}\right)^{\frac{1}{^{\alpha_s}}}\left(\sum_{r}p_{cr}d^s_{ckr}\right)^{\frac{(1-\alpha_s)}{^{-\alpha_s}}}}
\end{dmath}

\begin{dmath}
b^s_{ckr}=\frac{B_{s}\frac{p_{cr}d^s_{ckr}}{\sum_{r}p_{cr}d^s_{ckr}}\left( w^s_{ck}\right)^{\frac{1}{^{\alpha_s}}}\left(\sum_{r}p_{cr}d^s_{ckr}\right)^{\frac{(1-\alpha_s)}{^{-\alpha_s}}}}{\sum_{c}\sum_{k}\left( w^s_{ck}\right)^{\frac{1}{^{\alpha_s}}}\left(\sum_{r}p_{cr}d^s_{ckr}\right)^{\frac{(1-\alpha_s)}{^{-\alpha_s}}}}
\end{dmath}
\subsection{Market equilibrium}\label{proofmarketequi}

%\begin{dmath}
%\sum_{i}\frac{B_i}{(1-\alpha_s)}\log\left(\sum_{k} {u_{ik}}^{(1-\alpha_s)}\right)-\sum_{i}\sum_{k}\sum_{r}\lambda_{ikr}\left(u_{ik}-\frac{x_{ikr}}{d_{ikr}}\right)-\sum_{k}\sum_{r}p_{kr}\left(\sum_{i}x_{ikr}-1\right)
%\end{dmath}
%\begin{dmath}
%\frac{B_i{u_{ik}}^{(1-\alpha_s)-1}}{\sum_{k} {u_{ik}}^{(1-\alpha_s)}}-\sum_{r}\lambda_{ikr}=0
%\end{dmath}
%\begin{dmath}
%\frac{\lambda_{ikr}}{d_{ikr}}=p_{kr}
%\end{dmath}
%\begin{dmath}
%\lambda_{ikr}=p_{kr}d_{ikr}
%\end{dmath}
%\begin{dmath}
%\frac{B_i{u_{ik}}^{(1-\alpha_s)-1}}{\sum_{k} {u_{ik}}^{(1-\alpha_s)}}=\sum_{r}p_{kr}d_{ikr}
%\end{dmath}
%\begin{dmath}
%\frac{B_i{u_{ik}}^{(1-\alpha_s)}}{\sum_{k} {u_{ik}}^{(1-\alpha_s)}}=\sum_{r}p_{kr}x_{ikr}
%\end{dmath}
%\begin{dmath}
%B_{i}=\sum_{k}\sum_{r}p_{kr}x_{ikr}
%\end{dmath}
%  \begin{subequations}
%\begin{dmath}
%P_{SW}: \hspace{10pt}& \underset{{\bf x,u}}{\text{Maximize :}} 
%& & \sum_{s\in {\cal S}}B_s \log(U_s) \label{primal}\\
%& \text{subject to} & &  U_s=\left(\sum_{c\in\mathcal{C}}\sum_{k\in\mathcal{K}^s_c}w^s_{ck}\left(u^s_{ck}\right)^{(1-\alpha_s)}\right)^{\frac{1}{(1-\alpha_s)}} \quad \forall \; s\in {\cal S} \label{const1}\\
%&  & &u^s_{ck}\leq \frac{x^s_{ckr}}{d^s_{ckr}} \; (\lambda^s_{ckr}) \quad \forall \ s\in {\cal S},c\in {\cal C}_s,k\in {\cal K}_c,r\in {\cal R} \label{const2}\\
%& & &  \sum_{ s \in {\cal S}}\sum_{k\in\mathcal{K}^s_c}  x^s_{ckr}   \leq C_{cr},\; (p_{cr}) \quad \forall \ c\in {\cal C},r\in {\cal R} \label{const3}
%\end{dmath}\label{covexopti}
%  \end{subequations}
\begin{proof}
Consider the Lagrangian of EG problem
\begin{dmath}
L(u,x,\lambda,p)=\frac{B_{s}}{(1-\alpha_s)}\log\left(\sum_{c\in\mathcal{C}}\sum_{k\in\mathcal{K}^s_c}w^s_{ck}\left(u^s_{ck}\right)^{(1-\alpha_s)}\right)-\sum_{s}\sum_{c}\sum_{\mathcal{K}_{c}}\sum_{r}\lambda^s_{ckr}\left(u^s_{ck}-\frac{x^s_{ckr}}{d^s_{kr}}\right)\\
-\sum_{c}\sum_{r}p_{cr}\left(\sum_{s}\sum_{k}x^s_{ckr}-1\right)
\end{dmath}
After applying first order  KKT condition we have
\begin{align}
\frac{B_{s}w^s_{ck}\left(u^s_{ck}\right)^{-\alpha_s}}{\left(\sum_{c\in\mathcal{C}}\sum_{k\in\mathcal{K}^s_c}w^s_{ck}\left(u^s_{ck}\right)^{(1-\alpha_s)}\right)}-\sum_{r\in\mathcal{R}_c}\lambda^s_{ckr}=0\\
\forall s\in\mathcal{S},\forall c\in\mathcal{C},\forall k\in\mathcal{K}\nonumber
\end{align}
and
\begin{align}
&\frac{\lambda^s_{ckr}}{d^s_{kr}}-p_{cr}=0 \;\forall s\in\mathcal{S},\forall c\in\mathcal{C},\forall k\in\mathcal{K},r\in\mathcal{R}_
{c}\\
&u^s_{ck}=\frac{x^s_{ckr}}{d^s_{ckr}}  \;\forall s\in\mathcal{S},\forall c\in\mathcal{C},\forall k\in\mathcal{K},r\in\mathcal{R}_
{c}\\
&p_{cr}\left(\sum_{s}\sum_{k}x^s_{ckr}-1\right)=0\;\forall c\in\mathcal{C}\;\forall r\in\mathcal{R}_{c}\\
&\lambda^s_{ckr}\geq 0,\forall s\in\mathcal{S},\forall c\in\mathcal{C},\forall r\in\mathcal{R}_{c}\\
&p_{cr} \geq 0 \;\forall c\in\mathcal{C},\forall r\in\mathcal{R}_{c}
\end{align}
Form we have
\begin{dmath}
B_{s}\frac{w^s_{ck}\left(u^s_{ck}\right)^{-\alpha_s}}{\left(\sum_{c\in\mathcal{C}}\sum_{k\in\mathcal{K}^s_c}w^s_{ck}\left(u^s_{ck}\right)^{(1-\alpha_s)}\right)}=\sum_{r}p_{cr}d^s_{ckr}
\end{dmath}
From replacing $x^s_{ck}$ with $u^s_{ck}d^s_{ckr}$ we have
\begin{dmath}
B_{s}\frac{w^s_{ck}\left(u^s_{ck}\right)^{(1-\alpha_s)}}{\left(\sum_{c\in\mathcal{C}}\sum_{k\in\mathcal{K}^s_c}w^s_{ck}\left(u^s_{ck}\right)^{(1-\alpha_s)}\right)}=\sum_{r}p_{cr}x^s_{ckr}
\end{dmath}
summing over 
\begin{dmath}
B_{s}=\sum_{c}\sum_{k\in\mathcal{K}^s_c}\sum_{r}p_{cr}x^s_{ckr}
\end{dmath}
summing over the $s$ we get 
\begin{dmath}
\sum_{s}\sum_{c}\sum_{r}p_{cr}=1;
\end{dmath}
\end{proof}

\end{document}